\documentclass[journal,11pt,onecolumn]{IEEEtran}
\usepackage{amsmath,amsfonts}
\usepackage{amssymb}   
\usepackage{amsthm} 
\usepackage{algorithmic}
\usepackage{algorithm}
\usepackage{multirow}
\usepackage{array}
\usepackage{textcomp}
\usepackage{url}
\usepackage{verbatim}
\usepackage{graphicx}
\usepackage{subfigure}
\usepackage[justification=centering]{caption}

\def\BibTeX{{\rm B\kern-.05em{\sc i\kern-.025em b}\kern-.08em
    T\kern-.1667em\lower.7ex\hbox{E}\kern-.125emX}}
\usepackage{balance}
\usepackage{pifont}

\usepackage{lipsum}
\newtheorem{theorem}{Theorem}
\newtheorem{corollary}{Corollary}
\newtheorem{lemma}{Lemma}
\newtheorem{remark}{Remark}

\newcommand{\be}{\begin{equation}}
\newcommand{\ee}{\end{equation}}
\newcommand{\beq}{\begin{equation}}
\newcommand{\eeq}{\end{equation}}
\newcommand{\beqy}{\begin{eqnarray}}
\newcommand{\eeqy}{\end{eqnarray}}
\newcommand{\beqynn}{\begin{eqnarray*}}
\newcommand{\eeqynn}{\end{eqnarray*}}
\newcommand{\al}[1]{\begin{align}#1\end{align}}
\newcommand{\aln}[1]{\begin{align*}#1\end{align*}}

\newcommand{\ba}{\begin{array}}
\newcommand{\ea}{\end{array}}

\newcommand{\bmx}{\begin{bmatrix}}
\newcommand{\emx}{\end{bmatrix}}
\newcommand{\bsmx}{\left[\begin{smallmatrix}}
\newcommand{\esmx}{\end{smallmatrix}\right]}
\newcommand{\bmxc}[1]{\left[\begin{array}{@{}#1@{}}}
\newcommand{\emxc}{\end{array}\right]}

\newcommand{\bt}[1]{\begin{tabular}{#1}}
\newcommand{\et}{\end{tabular}}

\newcommand{\bc}{\begin{center}}
\newcommand{\ec}{\end{center}}
\newcommand{\ben}{\begin{enumerate}}
\newcommand{\een}{\end{enumerate}}
\newcommand{\bi}{\begin{itemize}}
\newcommand{\ei}{\end{itemize}}

\renewcommand{\top}{{\mkern-1.5mu\mathsf{T}}}

\newcommand{\erf}{\mathrm{erf}}

\DeclareMathOperator*{\argmin}{argmin}
\DeclareMathOperator*{\argmax}{argmax}

\newcommand{\Rbb}{{\mathbb{R}}}
\newcommand{\Zbb}{{\mathbb{Z}}}

\newcommand{\Rn}{\Rbb^{n}}

\newcommand{\Rnbn}{\Rbb^{n \times n}}

\newcommand{\Rmbm}{\Rbb^{m \times m}}
\newcommand{\Rmbn}{\Rbb^{m \times n}}

\newcommand{\Zn}{\Zbb^{n}}

\newcommand{\Znbn}{\Zbb^{n \times n}}

\newcommand{\sOB}{{\scriptscriptstyle \mathrm{OB}}}

\newcommand{\sBB}{{\scriptscriptstyle \mathrm{BB}}}

\newcommand{\sRB}{{\scriptscriptstyle \mathrm{RB}}}

\newcommand{\calX}{{\mathcal{X}}}

\newcommand{\A}{\mathbf{A}}

\newcommand{\D}{\mathbf{D}}

\newcommand{\F}{\mathbf{F}}
\newcommand{\G}{\mathbf{G}}

\newcommand{\I}{\mathbf{I}}

\renewcommand{\P}{\mathbf{P}}
\newcommand{\Q}{\mathbf{Q}}
\newcommand{\R}{\mathbf{R}}

\newcommand{\U}{\mathbf{U}}
\newcommand{\V}{\mathbf{V}}

\renewcommand{\a}{\mathbf{a}}

\renewcommand{\d}{\mathbf{d}}
\newcommand{\e}{\mathbf{e}}
\newcommand{\f}{\mathbf{f}}

\newcommand{\p}{\mathbf{p}}

\newcommand{\boldr}{\mathbf{r}}

\renewcommand{\v}{{\mathbf{v}}}

\newcommand{\x}{{\mathbf{x}}}
\newcommand{\y}{{\mathbf{y}}}
\newcommand{\z}{{\mathbf{z}}}

\newcommand{\0}{{\boldsymbol{0}}}

\newcommand{\br}{{\bar{r}}}

\newcommand{\bE}{{\bar{E}}}

\newcommand{\blam}{{\bar{\lambda}}}
\newcommand{\bomega}{{\bar{\omega}}}

\newcommand{\hr}{\hat{r}}

\newcommand{\hbR}{\hat{\R}}

\newcommand{\hbv}{{\hat{\v}}}

\newcommand{\hby}{{\hat{\y}}}

\newcommand{\tv}{\tilde{v}}
\newcommand{\ty}{\tilde{y}}

\newcommand{\tbv}{{\tilde{\v}}}

\newcommand{\tby}{{\tilde{\y}}}

\begin{document}
\title{Success probability of the $L_0$-regularized box-constrained Babai point and column permutation strategies}
\author{%
\IEEEauthorblockN{Xiao-Wen Chang and Yingzi Xu}
\IEEEauthorblockA{School of Computer Science, McGill University \\
  3480 University Street, Montreal, H3A 0E9, Canada \\
  Email: xiao-wen.chang@mcgill.ca,  yingzi.xu@mail.mcgill.ca
  }
}
\maketitle

\begin{abstract}
We consider the success probability 
of the $L_0$-regularized box-constrained Babai point, 
which is a suboptimal solution to the $L_0$-regularized 
box-constrained integer least squares problem
and can be used for MIMO detection.
First, we derive formulas for the success probability 
of both $L_0$-regularized and unregularized box-constrained Babai points. 
Then we investigate the properties of the $L_0$-regularized box-constrained Babai point, 
including the optimality of the regularization parameter, 
the monotonicity of its success probability,
and the monotonicity of the ratio of the two success probabilities.
A bound on the success probability of the $L_0$-regularized Babai point
is derived.
After that, we analyze the effect of the LLL-P  permutation strategy 
on the success probability of the $L_0$-regularized Babai point. 
Then we propose some success probability based 
column permutation strategies to increase the success probability 
of the $L_0$-regularized box-constrained Babai point. 
Finally, we present numerical tests to confirm our theoretical results and 
to show the advantage of the $L_0$ regularization
and the effectiveness of the proposed column permutation algorithms 
compared to existing strategies. 
\end{abstract}

\begin{IEEEkeywords}
Detection, $L_0$-regularization, integer least squares, Babai point, sparse signals, success probability, column permutations
\end{IEEEkeywords}

\section{Introduction}
In many applications, the unknown parameter vector $\x^\ast$ and the observation vector $\y$ satisfy  
the following linear relation:
\be \label{eq:lm}
\y=\A\x^\ast+\v,  \ \ \v \sim N(\0, \sigma^2 \I),
\ee 
where $\A\in \Rmbn$ is a model matrix and $\v$ is an $m$ dimensional noise vector.
In this paper, we consider the case that $\x^\ast$ is random, 
and given integer $M\geq 1$, the entries of $\x^\ast$ are independently distributed over the set
\be \label{eq:setX}
 {\cal X} =\{0\}\cup \{-2M+1, -2M+3, \ldots, -1,  1, \ldots, 2M-3, 2M-1 \},
\ee
and the elements in the set are chosen independently
with probability
\be \label{eq:xprob}
\Pr(x^\ast_k=i)= \left\{ \ba{ll} p/(2M),  & i  \in {\cal X},  i \neq 0  \\
 1- p, & i = 0, \ea \right.
\ee
where  $p$ is a positive constant satisfying
\be
p \leq \frac{2M}{2M+1}. \label{eq:zeropr}
\ee 
The inequality \eqref{eq:zeropr} is equivalent to 
\be \label{eq:pinequ}
p/(2M) \leq 1-p,
\ee
i.e., the probability for $x_k^*$ to take any nonzero integer in the set $\calX$ is the same,
and it is smaller than or equal to the probability for $x_k^*$ to take 0.
If $p$ is small,  $\x^\ast$ is sparse. 
Note that all numbers except zero in ${\cal X}$ are odd numbers. 
This model has applications in multiple user detection in code-division multiple
access (CDMA)  communications, e.g., see \cite{ZhuG11},
in which ${\cal X}$ is a $2M$-ary constellation.
When $x_k^\ast$ takes 0, the $k$-th user is said to be inactive.

To recover $\x^\ast$ from the observation $\y$, one can apply the maximum a posteriori (MAP)  estimation method, which solves 
\be \label{eq:map}
\max_{\x \in {\cal X}^n}  \ln \left(\varphi ( \y |\x^\ast=\x) \Pr (\x^\ast =\x)\right),
\ee
where $\varphi( \y |\x^\ast=\x)$ is the probability density function of $\y$ when $\x^\ast=\x$
(here for the sake of notational simplicity, $\y$ is also used for its realized value)
and $\Pr (\x^\ast =\x)$ is the probability that $\x^\ast=\x$ for $\x\in {\cal X}^n$.
It is shown in \cite{ZhuG11} that \eqref{eq:map} is equivalent to
 the following regularized optimization problem:
\be \label{eq:l0rbils}
\min_{\x \in {\cal X}^n} \frac{1}{2} \|\y-\A\x\|_2^2  +  \lambda \|\x\|_0,
\ee
where  the $L_0$-``norm'' $\|\x\|_0$ is the number of nonzero elements of $\x$, and the regularization parameter  $\lambda = \lambda^*$ with 
\be \label{eq:lambda}
\lambda^*:= \sigma^2 \ln \frac{1-p}{p/(2M)} \geq 0.
\ee

In \cite{ZhuG11}, the inequality \eqref{eq:zeropr} is strict, then $\lambda^*>0$.
We allow \eqref{eq:zeropr} to be an equality so that $\lambda^*$ can be zero. 
This gives us more flexibility.
To make some of our results more applicable, in this paper we assume that the $\lambda$ in \eqref{eq:l0rbils} 
is only  a given nonnegative parameter, unless we state explicitly that it is defined by \eqref{eq:lambda}.
This unifies some results for $\lambda=\lambda^*$ and $\lambda=0$.
We refer to \eqref{eq:l0rbils} as the $L_0$-``norm'' regularized ``box'' constrained integer least squares ($L_0$-RBILS) problem.
(Although  the constraint set ${\cal X}^n$ is not a regular box, its geometric shape is still a box.)

If no \textit{a priori} information about $\x^\ast$ in \eqref{eq:lm} is available, to estimate  $\x^\ast$, one applies the maximum likelihood (ML) 
estimation method,
leading to the box constrained integer least squares (BILS) problem:
\be \label{eq:bils}
\min_{\x \in {\cal X}^n} \frac{1}{2} \|\y-\A\x\|_2^2.
\ee
Note that this problem is just \eqref{eq:l0rbils} without the regularization term.
If $\x^\ast$ is uniformly distributed over ${\cal X}^n$, then  \eqref{eq:zeropr} is an equality, and then $\lambda^*=0$
and \eqref{eq:l0rbils} with $\lambda=\lambda^*$ just becomes \eqref{eq:bils}. 
We can regard \eqref{eq:bils} as a special case of \eqref{eq:l0rbils}.
If we replace the constraint box ${\cal X}$ by $\Zn$, \eqref{eq:bils} becomes the ordinary integer least squares (OILS) problem
\be \label{eq:ils}
\min_{\x \in \Zn} \frac{1}{2} \|\y-\A\x\|_2^2.
\ee


The OILS problem \eqref{eq:ils} is  NP-hard  \cite{Mic01,Van81} and there are a few approaches 
to  finding the optimal  solution (see, e.g., \cite{HanPS11,AnjCK14}).
The widely used approach in applications is the enumeration approach originally proposed in \cite{FinP85}
and later improved in \cite{SchE94},
which enumerates integer points in an ellipsoid and is called sphere decoding in communications.
This approach has been extended to constrained problems \cite{DamGC03,ChaH08,ChaG10}
and to the $L_0$-RBILS problem  \cite{ZhuG11,BarV14,MonBD13}.
To make the search faster, typically the original problem is transformed to a new problem
by lattice reduction or column permutations. 
Some examples of commonly used strategies include 
LLL reduction \cite{LenLL82}, 
V-BLAST \cite{FosGVW99} and sorted QR decomposition (SQRD) \cite{WubBRKK01}.

Since an ILS problem is NP-hard, often a suboptimal solution instead of the optimal one is found in some applications.
One often used suboptimal solution to the ILS problem is the Babai point \cite{Bab86}, 
which was originally introduced for the OILS problem \eqref{eq:ils} but
has been extended to the BILS problem  \cite{DamGC03,ChaH08} and to
the $L_0$-RBILS problem \cite{ZhuG11}.
It is the first integer point found by the Schnorr-Euchner type search methods
when the radius of the initial search ellipsoid is set to $\infty$
\cite{SchE94}.
To see how good the Babai point is as a detector,
one can use {\em success probability} (SP) as a measure, which is the probability that the detector
is equal to the true parameter vector. 
If the SP of the Babai point is high, one will not need to try to find 
the optimal detector.

In this paper we conduct systematic studies of the SP of the $L_0$-regularized Babai point $\x^\sRB$ (here the superscripts RB stands for ``regularized Babai'') for the $L_0$-RBILS problem \eqref{eq:l0rbils}.
The main contributions are summarized as follows.
\ben 
\item
We establish a formula for  the SP of the $L_0$-regularized
Babai point $\x^\sRB$ 
and a formula for the SP of the (unregularized) Babai point $\x^\sBB$ (here the superscripts BB standard for ``box-constrained Babai'')
for the BILS problem \eqref{eq:bils} in Section \ref{sec:spbabai}.

\item 
The properties of the SP of $\x^\sRB$ are investigated in Section \ref{sec:bproperties}.
We show that the regularization parameter $\lambda$ defined as $\lambda^*$ in 
\eqref{eq:lambda} makes the SP of $\x^\sRB$ the highest among all choices of $\lambda\geq 0$,
the SP of $\x^\sRB$ is monotonically increasing with respect to
the diagonal elements of the R-factor of the QR factorization of $\A$,
and the ratio of the SP of $\x^\sRB$ to SP of $\x^\sBB$ is monotonically decreasing
with respect to $p\in (0,2M/(2M+1))$.
We also derive a bound on the SP of $\x^\sRB$, 
which can be either a lower bound or an upper bound, depending on the conditions.

\item 
We analyse the effect of the column permutation strategy LLL-P on the SP of $\x^\sRB$
in Section \ref{sec:permut}.
Specifically we show that the LLL-P strategy is guaranteed 
to increase or decrease the SP of $\x^\sRB$ under certain conditions.

\item 
Based on the formula of the SP of $\x^\sRB$, we propose 
three column permutation strategies 
to increase the SP of the Babai point in Section \ref{sec:spgalg}.
The LSP strategy uses the local information of the SP,
which turns out to be the well known V-BLAST strategy,
while the GSP strategy uses the global information of the SP.
The MSP strategy is a mixed strategy based on LSP and GSP.
\een

The $L_0$-regularized box-constrained Babai point $\x^\sRB$ considered 
in this paper can be regarded as an extension of 
the ordinary Babai point and the box-constrained Babai point,
whose SPs are studied in \cite{ChaWX13} and \cite{WenC17}, respectively.
Some of the results in Section \ref{sec:spbabai}-\ref{sec:bproperties} resemble the results in \cite{ChaWX13} and \cite{WenC17}, which study the Babai point for OILS and BILS problems respectively. 
Table \ref{t:paper_comparison} gives a brief comparison between 
this work and \cite{ChaWX13,WenC17}.

\begin{table*}[]
\centering
\caption{Comparisons of the works of this paper and \cite{ChaWX13, WenC17}}
\label{t:paper_comparison}
\begin{tabular} 
{|p{0.15\linewidth} || p{0.25\linewidth} | p{0.25\linewidth}| p{0.25\linewidth}|}  
\hline
Paper & \multicolumn{1}{c|}{ \cite{ChaWX13}} 
  & \multicolumn{1}{c|}{ \cite{WenC17}}
 & \multicolumn{1}{c|}{This work}    \\ \hline
True parameter vector $\x^\ast$    & {$\x^\ast\in \Zn$}   
  & $\x^\ast\in \Zn$ uniformally distributed over 
  $[l_1,u_1]\times \cdots\times [l_n,u_n]$
  & $\x^\ast \in \calX^n$ is sparse with distribution \eqref{eq:xprob}  \\ \hline
Detector & Ordinary Babai point 
& Box-constrained Babai point &
$L_0$-regularized box-constrained Babai point \\ \hline
SP formula       &  \multicolumn{1}{c|}  \checkmark  &  
\multicolumn{1}{c|}\checkmark   &  \multicolumn{1}{c|} \checkmark   \\ \hline
Properties of the SP of the detector &
\hspace{18mm}  \ding{55} &
The SP is higher than that of the ordinary Babai point. & 
The SP reaches the maximum when $\lambda=\lambda^*$. \newline
The SP increases w.r.t. all  $r_{kk}$. \newline
The ratio of SP of $\x^\sRB$ to SP of $\x^\sBB$ decreases w.r.t. $p$.
\\ \hline
Bound on SP 
& A lower bound and upper bounds that hold at all times.   
& 
Lower/upper bound under some conditions. 
& 
Lower/upper bound  under some conditions.  \\ \hline
Effect of reduction or permutations  
& LLL  increases SP 
and reduces the cost of sphere decoding.  
& LLL-P increases/decreases SP under some conditions 
& LLL-P increases/decreases SP under some conditions \\ \hline
Permutation algorithms  &  \hspace{18mm}  \ding{55}
& 
 \hspace{18mm}  \ding{55} &  
  \hspace{18mm} \checkmark      \\ \hline
\end{tabular}
\end{table*}


Here we would like to point out that the discrete sparse signal detection or recovery has recently attracted much attention \cite{ChoSDRK17}, and various methods have been proposed, see, e.g.,  \cite{ZhuG11,
SchD11,SchD12,KnoMBWPD13,PantLA14, BarV14,LeeCS16, AhnSL17,WenPYYL17,SelF17,SouL17,CirML18,FukNS19,Hay20, LanPST20, CheZCZ21}. 
Comparisons of the $L_0$-RBILS detector with other detectors are beyond the scope of this paper.

{\bf Notation.} We use boldface lowercase and uppercase letters to denote vectors and matrices, respectively. 
For a column vector $\x \in \Rn$,  
$x_i$ denotes the $i$th entry of $\x$, and
$\x_{ij}$
denotes the subvector composed of
elements of $\x$ with indices from $i$ to $j$. 
For a matrix $\A \in \Rmbn$,  
$\a_i$ denotes the $i$th column of $\A$, and
$\A_{ij,kl}$
denotes the submatrix containing all the elements of $\A$ whose row indices are
from $i$ to $j$ and column indices are from $k$ to $l$.
For a scalar $x \in \Rbb$,
$\lfloor x \rceil $ denotes the nearest integer to $x$
(if there is a tie, the one with smaller magnitude is chosen),
and similarly, 
$\lfloor x \rceil _{\calX}$ denotes the integer in 
the set $\calX$ of integers nearest to $x$.
We use $\I$ to denote an identity matrix  and $\e_k$ to
denote its $k$-th column. 
If a random vector $\x$ is normally distributed with zero mean and covariance matrix $\sigma^2\I$,
we write $\x \sim N(\0,\sigma^2\I)$.

\section{Preliminaries}
In this section, we first transform the $L_0$-RBILS problem \eqref{eq:l0rbils} to an equivalent one
by the QR factorization of $\A$ and define the corresponding Babai point.
Then we give some background for the column permutation strategies 
to be discussed in this paper.

\subsection{The Babai points}
We first transform  \eqref{eq:l0rbils} to an equivalent new problem by an orthogonal transformation.
Let the QR factorization of $\A$ be
\be \label{eq:qr}
\A=[\Q_1,\Q_2]\bmx \R \\ \0\emx = \Q_1 \R,
\ee
where $[\Q_1,\Q_2] \in \Rmbm$ is orthogonal and $\R\in \Rnbn$ is upper triangular with {\em positive} diagonal entries. 
Given the matrix $\A$, such $\R$ is unique and can be found easily and efficiently using algorithms such as the Householder transformations. 

With $\tby=\Q_1^\top \y$ and $\tbv=\Q_1^\top \v$, it is easy to show the model \eqref{eq:lm} can be transformed to
\be \label{eq:rlm}
\tby = \R\x^\ast + \tbv, \  \ \tbv \sim N(\0,\sigma^2\I),
\ee
and the $L_0$-RBILS problem \eqref{eq:l0rbils} is  equivalent  to
\be \label{eq:l0rbils1}
\min_{\x \in {\cal X}^n} \frac{1}{2} \|\tby-\R\x\|_2^2  +  \lambda \|\x\|_0.
\ee

Here we define the $L_0$-regularized Babai point $\x^\sRB=[x_1^\sRB, \ldots, x_n^\sRB]^\top $ 
corresponding  to \eqref{eq:l0rbils1}.
Note that the objective function in \eqref{eq:l0rbils1} can be rewritten as
$$
\min_{\x \in {\cal X}^n}
\sum_{k=1}^n \bigg(\frac{1}{2} \Big(\ty_k- r_{kk}x_k-\sum_{j=k+1}^n r_{kj} x_j  \Big)^2  +  \lambda \|x_k\|_0\bigg). 
$$
Suppose that $x_n^\sRB, \ldots, x_{k+1}^\sRB$ have been defined and 
we would like to define  $x_k^\sRB$ as an estimator of $x_k^*$.
The idea is to choose $x_k^\sRB$ to be the solution of the optimization problem:
\be \label{eq:opk}
\min_{x_k \in {\cal X} } \Big\{f_k(x_k) :
= \frac{1}{2} \Big(\ty_k- r_{kk}x_k-\sum_{j=k+1}^n r_{kj} x_j^\sRB \Big)^2  +  \lambda \|x_k\|_0\Big\},
\ee
where $\sum_{k+1}^n \cdot = 0$ if $k=n$.
To simplify the objective function in \eqref{eq:opk},
write 
\be \label{eq:ck}
c_k:=\Big(\ty_k-\sum_{j=k+1}^n r_{kj} x_j^\sRB \Big)/r_{kk}.
\ee
Then we have
\be \label{eq:fxk}
f_k(x_k)= \frac{1}{2}r_{kk}^2(x_k-c_k)^2 +  \lambda \|x_k\|_0.
\ee
From \eqref{eq:fxk} we observe that the solution to \eqref{eq:opk} is either 0 or $\lfloor c_k \rceil_{\cal X}$.
Of course the two are identical if $\lfloor c_k \rceil_{\cal X}=0$. 
To see which one is the solution.
Define
\be \label{eq:gck}
g_k:= \frac{1}{2} r_{kk}^2\lfloor c_k \rceil_{\cal X}^2-  r_{kk}^2 \lfloor c_k \rceil_{\cal X}  c_k + \lambda.
\ee
Observe that $g_k = f_k(\lfloor c_k \rceil_{\cal X}) - f_k(0) $ when $\lfloor c_k \rceil_{\cal X} \neq 0$, and $g_k \geq 0$ when $\lfloor c_k \rceil_{\cal X} = 0$.
Thus, the solution to \eqref{eq:opk} is
\be \label{eq:xkrb}
x_k^\sRB = 
\left\{\begin{array}{@{}cl}
  0, & \text{if } g_k  \geq  0, \\
  \lfloor c_k \rceil_{\cal X}, & \text{if } g_k  <  0.
\end{array} \right.
\ee
If we set $\lambda=0$ in \eqref{eq:l0rbils},    $\x^\sRB$ just becomes 
the (unregularized) box-constrained Babai point $\x^\sBB$ of the BILS problem \eqref{eq:bils}:
$$
x_k^\sBB =  \Big\lfloor \Big(\ty_k-\sum_{j=k+1}^n r_{kj} x_j^\sBB \Big)/r_{kk} \Big\rceil_{\cal X}, 
\ \ k = n,n-1,\ldots,1,
$$
which is an extension of the ordinary Babai point for the OILS problem \eqref{eq:ils}:
$$
x_k^\sOB =  \Big\lfloor \Big(\ty_k-\sum_{j=k+1}^n r_{kj} x_j^\sOB \Big)/r_{kk} \Big\rceil, \ \ k = n,n-1,\ldots,1.
$$

Here  the $L_0$-regularized Babai point $\x^\sRB$ is a further extension of the ordinary Babai point $\x^\sOB$.
In \cite{ZhuG11} an equivalent point is defined, and it is referred to as the decision-directed detector.

\subsection{Column permutation strategies} \label{sec:perm}
In OILS problems, a commonly used lattice reduction strategy is LLL reduction, which consists of two types of operations: size reduction and column permutation. 
In BILS problems, applying size reduction would make the transformed constraints difficult to handle in the search process. 
However we can still use the permutation strategy in the reduction process, 
referred to as LLL-P \cite{ChaWX13}.
The resulting matrix $\R$ satisfies the Lov\'{a}sz condition 
\begin{equation} \label{eq:lovasz}
    \delta\, r_{k-1,k-1}^2 \leq r_{k-1,k}^2 + r_{kk}^2, \quad k=2,3,\ldots,n,
\end{equation}
where $\delta\in (\frac{1}{4},1]$ is a parameter. 

The reduction process works with the upper triangular matrix $\R$ obtained by the QR factorization \eqref{eq:qr}.
If the Lov\'{a}sz condition \eqref{eq:lovasz} is violated 
for some $k$, then the algorithm interchanges columns $k-1$ and $k$ of $\R$
and uses a Givens rotation to restore the upper-triangular structure of $\R$:
\begin{equation}\label{eq:llltri}
    \hat{\R} := \G^\top _{k-1,k}\R\P_{k-1,k},
\end{equation}
where $\G_{k-1,k}$ is a Givens rotation matrix and $\P_{k-1,k}$ 
is a permutation matrix. 
Then it is easy to show that 
\begin{equation}\label{eq:lllprop}
\begin{aligned}
    \hr_{k-1,k-1}^2 &= r_{k-1,k}^2 + r_{kk}^2,   \\
    \hr_{k-1,k}^2 + \hr_{kk}^2 &= r_{k-1,k-1}^2, \\
    \hr_{k-1,k-1} \hr_{kk} &= r_{k-1,k-1}r_{kk},
\end{aligned}
\end{equation}
which will be used in our later analysis.
Then it is easy to see that $\hbR$ satisfies the Lov\'{a}sz condition for that $k$.

For the sake of readability, we give the reduction algorithm which implements the LLL-P
strategy in Algorithm \ref{alg:lll-p}.

\begin{algorithm}
\caption{LLL-P}
\label{alg:lll-p}
\begin{algorithmic}[1]
\STATE  Initialize $\p = [1,2,\ldots,n],k=2$

\WHILE{$k\leq n$}
    \IF{$\delta\, r_{k-1,k-1}^2 > r_{k-1,k}^2 + r_{kk}^2$}
    \STATE Interchange $p_{k-1}$ and $p_k$ \\
    and interchange $\boldr_{k-1}$ and $\boldr_k$.
    \STATE Triangularize $\R$ by Givens rotations: $\R = \G^\top _{k-1,k}\R$ 
    \STATE $k=k-1$ when $k>2$
    \ELSE
    \STATE $k=k+1$
    \ENDIF
\ENDWHILE
\end{algorithmic}
\end{algorithm}

The V-BLAST strategy determines the columns from the last to the first. 
Like LLL-P, it starts with the factor $\R$ of the QR factorization of $\A$ 
(see \eqref{eq:qr}). 
For $k=n,n-1,\ldots,2$, the $k$-th column is chosen from 
the first $k$ columns in the permuted matrix $\R$ 
such that the new $r_{kk}$ is the largest. 
An efficient implementation of the V-BLAST strategy 
is given in Algorithm \ref{alg:vblast}
(see \cite{ChaP07}), which costs $O(n^3)$ flops.

\begin{algorithm}
\caption{V-BLAST}\label{alg:vblast}
\begin{algorithmic}[1]
\STATE  Initialize $\p = [1,2,\ldots,n]$, $\F = \R^{-\top}$.
\FOR{$k=n,n-1,\ldots,2$}
    \STATE Compute $\d$: 
    $d_j = \| (\f_j)_{jk} \|_2$
    for $j=1,2,\ldots,k$.
    \STATE Set $j_k = \argmin_{1\leq j\leq k} d_j$.
    \STATE Interchange $p_k$ and $p_{j_k}$, and $\boldr_k$ and $\boldr_{j_k}$. 
    \STATE Triangularize $\R$ by Givens rotations: $\R := \G^\top \R$. 
    \STATE Interchange $\f_k$ and $\f_{j_k}$, and update $\F = \G^\top \F$.
\ENDFOR
\end{algorithmic}
\end{algorithm}


On the other hand, the SQRD strategy acts on $\A$ directly 
and the original implementation  is built on the basis of the modified Gram-Schmidt process 
for the QR factorization,
which determines the columns of the permuted $\R$ from the first to the last. 
For $k=1,2,\ldots,n-1$, the $k$-th column of the permuted $\R$ is decided by 
exchanging column $k$ with the column with the smallest norm 
in the last $n-k+1$ columns in $\A$, then we orthogonalize $\A$ and update $\R$
the same way as in the modified Gram-Schmidt process.
This gives us the smallest $r_{kk}$ among the available candidates. 
A more efficient and numerically stable implementation is to use 
the Householder transformations, very similar to 
the QR factorization with column pivoting by Householder
transformations, see, e.g., \cite[p278]{GolV13}.

Regardless of the reduction strategy used,
after reduction we obtain the upper triangular matrix. 
We denote it by $\hat{\R}$ in order to distinguish it 
from $\R$ in \eqref{eq:qr} and can write
\begin{equation}\label{eq:pqr}
    \A\P = \hat{\Q}_1\hat{\R},
\end{equation}
where $\P \in \Znbn$ is a permutation matrix,
$\hat{\Q}_1\in \Rbb^{m\times n}$ is column orthonormal.
We then define
\begin{equation*}
        \hby = \hat{\Q}_1^\top  \y, \quad \z^\ast = \P^\top \x^\ast, \quad \hbv = \hat{\Q}_1^\top \v, \quad \z = \P^\top \x.
\end{equation*}
Then the linear model \eqref{eq:lm} is transformed to 
\begin{equation}\label{eq:permuted-lm}
    \hby = \hat{\R}\z^\ast+\hbv, \quad \hbv \sim N(\mathbf{0}, \sigma^2\I)
\end{equation}
and the $L_0$-RBILS problem \eqref{eq:l0rbils} is equivalent to 
\begin{equation}\label{eq:l0rbils-p}
    \min_{\z\in {\cal X}^n} \frac{1}{2}\|\hby - \hat{\R}\z\|_2^2 + \lambda\|\z\|_0.
\end{equation}
Furthermore, we let the $L_0$-regularized Babai point $\z^\sRB = [z_1^\sRB, \ldots, z_n^\sRB]^\top $ corresponding to \eqref{eq:l0rbils-p}.
Then we define $\x^\sRB := \P\z^\sRB$.

\section{SP of the $L_0$-regularized Babai point $\x^\sRB$}\label{sec:spbabai}

We would like to see how good the Babai detector $\x^\sRB$ is.
One measure is the SP $\Pr(\x^\sRB=\x^\ast)$.
In this section  we derive a formula for it.

In our probability analysis, we need to use the well-known error function and related probability formulas.
The  error function is defined as 
\be \label{eq:ef}
\erf(\zeta) =  \frac{2}{\sqrt{\pi}} \int_0^\zeta \exp(-t^2) dt.
\ee
Note that $\erf(\infty)=1$.
Given $a$ and $b$, 
if  $x \sim N(0,\sigma^2)$, then
\be 
\Pr(x \leq a)  
= \frac{1}{2}\left( 1+ \mbox{erf}\left(\frac{a}{\sqrt{2}\sigma}\right) \right), \label{eq:prxla}
\ee 
\be
\Pr(x \geq a) 
= \frac{1}{2}\left( 1- \mbox{erf}\left(\frac{a}{\sqrt{2}\sigma}\right) \right), \label{eq:prxga}
\ee
\begin{align}
 \Pr(a \leq x \leq b)
  &  =  \frac{1}{2}\left(\! \mbox{erf}\left(\frac{(b-a)_+ +a }{\sqrt{2}\sigma}\right)
\!\!-\!\!  \mbox{erf}\left(\frac{a}{\sqrt{2}\sigma}\right) \right) \nonumber \\
 &  =  \frac{1}{2}\left(\! \mbox{erf}\left(\frac{b}{\sqrt{2}\sigma}\right)
\!\!-\!\!   \mbox{erf}\left(\frac{b-(b-a)_+ }{\sqrt{2}\sigma}\right) \right).
\label{eq:praxb}  
\end{align}
where $(b-a)_+ = b-a$ when $b-a>0$, otherwise $(b-a)_+ = 0$.

From \eqref{eq:rlm}, we have
$$
\ty_k = r_{kk}x^\ast_k + \sum_{j=k+1}^n r_{kj} x^\ast_j + \tv_k,  \ \ \tv_k \sim N(0,\sigma^2).
$$
Substituting it into \eqref{eq:ck}, we obtain
$$
c_k = x^\ast_k +  \sum_{j=k+1}^n \frac{r_{kj}}{r_{kk}} (x^\ast_j - x_j^\sRB) + \frac{\tv_k}{r_{kk}}.
$$
If $x_j^\sRB=x^\ast_j$ for $j=k+1,k+2,\ldots,n$, then 
\be \label{eq:ckdis}
c_k-x^\ast_k \sim N(0, \sigma^2/r_{kk}^2).
\ee
The key to finding $\Pr(\x^\sRB=\x^\ast)$ is to find the the range of $c_k-x_k^\ast$ 
for which $x_k^\sRB=x_k^\ast$ for $k=n,n-1,\ldots,1$.
When $x_k^\sRB=x_k^\ast$, $c_k-x_k^\ast=c_k-x_k^\sRB$.
Thus we will find the range of $c_k-x_k^\sRB$ for which $x_k^\sRB=x_k^\ast$.

Since the elements of the constraint set ${\cal X}$ are symmetric with respect to 0
and the distribution of $x_k^\ast$ is symmetric with respect to 0, 
we only focus on the nonnegative case  that  $x_k^\ast \geq 0$ and $c_k \geq 0$  in our following analysis,
and we can simply double the relevant probabilities when we also take the nonpositive case into account.

If $\lfloor c_k \rceil_{\cal X}=0$, i.e., $0\leq c_k \leq 1/2$, obviously in this case from \eqref{eq:xkrb} $x_k^\sRB=0$. 
In the rest of this paragraph we assume that $\lfloor c_k \rceil_{\cal X}> 0$.
From \eqref{eq:gck},
\be \label{eq:gck1}
g_k = - \ r_{kk}^2\lfloor c_k \rceil_{\cal X}\left(  c_k-\lfloor c_k \rceil_{\cal X} 
- \frac{ \lambda }{r_{kk}^2\lfloor c_k \rceil_{\cal X}} + \frac{1}{2}\lfloor c_k \rceil_{\cal X}\right).
\ee
 Then from  \eqref{eq:gck}, \eqref{eq:xkrb}, and  \eqref{eq:gck1},
\be\label{eq:xksoln}
\begin{split}
& x_k^\sRB=  \lfloor c_k \rceil_{\cal X}  \Leftrightarrow g_k < 0   
  \Leftrightarrow   c_k- \lfloor c_k \rceil_{\cal X}  
>  \frac{ \lambda }{r_{kk}^2\lfloor c_k \rceil_{\cal X}} - \frac{1}{2}\lfloor c_k \rceil_{\cal X},  \\
& x_k^\sRB= 0  \Leftrightarrow
g_k \geq 0  
  \Leftrightarrow   c_k- \lfloor c_k \rceil_{\cal X} 
\leq  \frac{ \lambda }{r_{kk}^2\lfloor c_k \rceil_{\cal X}} - \frac{1}{2}\lfloor c_k \rceil_{\cal X}.
\end{split}
\ee
In addition, from the definition of  $\lfloor c_k \rceil_{\cal X}$,   if $M\geq 2$, $c_k-\lfloor c_k \rceil_{\cal X}$ has to satisfy
\be \label{eq:ckrange}
\begin{cases}
 -1/2 < c_k - \lfloor c_k \rceil_{\cal X} \leq 1,  &   \lfloor c_k \rceil_{\cal X}=1, \\
 -1 < c_k - \lfloor c_k \rceil_{\cal X} \leq 1,  & \lfloor c_k \rceil_{\cal X}=2j-1, 
 j=2,\ldots,M-1, \\
 -1 < c_k - \lfloor c_k \rceil_{\cal X}, &  \lfloor c_k \rceil_{\cal X}=2M-1,
\end{cases}
\ee
and if $M=1$, it has to satisfy
\be \label{eq:ckrange1}
-1/2 < c_k - \lfloor c_k \rceil_{\cal X}.
\ee

To derive the formula for the SP of $\x^\sRB$, 
we define 
\be  \label{eq:alphakj}
\alpha_k^{(j)}  := \frac{\lambda}{r_{kk}^2(2j-1)}-\frac{1}{2} (2j-1), \ \ j=1,\ldots,M,
\ee
which is the rightmost hand side of \eqref{eq:xksoln}
when $\lfloor c_k \rceil_{\cal X}=2j-1$. 
Additionally, we need the following lemma.

\begin{lemma} \label{le:akj}
Suppose that $\lambda\geq 0$ and  $M\geq 2$.
Let $\alpha_k^{(j)}$ be defined by \eqref{eq:alphakj}. 
Then 
$$
\alpha_k^{(1)} \geq -1/2,
$$
and $\alpha_k^{(j)}$ is strictly decreasing when $j$ is increasing from 1 to $M$.
Let 
\be \label{eq:j0}
j_k := \underset{\alpha_k^{(j)} \geq -1}{\max} j.
\ee
If $j_k \geq 2$, then
\be \label{eq:akjg1}
\alpha_k^{(j)} \geq 1, \ \  j=1:j_k-1,
\ee
and if $j_k \leq M-1$, then
\be \label{eq:akjl1}
\alpha_k^{(j_k)} \leq 1, \ \ \alpha_k^{(j)} < -1, \ \  j=j_k+1,\ldots, M.
\ee
\end{lemma}

\begin{proof} See Appendix \ref{app:le:akj}.
\end{proof}

Based on the discussion above and Lemma \ref{le:akj}, 
we list the range of $c_k -x_k^\sRB $ for different cases of $\lfloor c_k \rceil_\calX$ in Table \ref{t:ckxk} for $M\geq 2$ and in Table \ref{t:ckxkm1}
for $M=1$, in order to assist the derivation of the SP of $\x^\sRB$ in the following theorem. 

\begin{table*}[h] 
\caption{Range of $c_k-x_k^\sRB$ for $\lambda\neq 0$ and $M\geq 2$}\label{t:ckxk}
\centering
\begin{tabular}{|c|c|c|c|c|} \hline
Case & Range of $c_k$ & $\lfloor c_k \rceil_{\cal X}$ 
  & $x_k^\sRB$ & Range of $c_k -x_k^\sRB$  \\ \hline \hline
0 & $ 0 \leq c_k \leq 1/2$ & 
   0    &   0  & 
   $ 0 \leq c_k  \leq 1/2$  \\ \hline
\multirow{2}{*}{1} & \multirow{2}{*}{$1/2 < c_k \leq 2$} & \multirow{2}{*}{ $1$} & 
      $1$ &  $ \alpha_k^{(1)}=\max\{ \alpha_k^{(1)},-1/2\} < c_k - 1 \leq1$   \\ \cline{4-5}
     &   &  & $0$ &$ 1/2 < c_k \leq 1+\min \{ \alpha_k^{(1)}, 1\}$   \\ \hline
$j$ & \multirow{2}{*}{$2j-2 < c_k \leq 2j$} & 
\multirow{2}{*}{ $2j-1$ } & 
      $2j-1$ &  $  \max\{\alpha_k^{(j)},-1 \} < c_k -(2j-1) \leq1$ \\ \cline{4-5}
($j \!=\!2\!:\!M\!-\!1$) &  & & 
    $0 $ &  $2j-2 < c_k \leq  2j-1+\min \{ \alpha_k^{(j)}, 1\} $ \\ \hline
\multirow{2}{*}{ $M$} &   \multirow{2}{*}{ $2M-2\leq c_k  $}  & 
\multirow{2}{*}{  $2M-1$} &   
     $2M-1$ & $\max\{\alpha_k^{(M)},-1 \} < c_k -(2M-1)$     \\ \cline{4-5}
    &  & &  
    0 & $2M-2  < c_k \leq 2M-1 +\alpha_k^{(M)}$ \\ \hline
\end{tabular}
\end{table*}

\begin{table*}[h] 
\caption{Range of $c_k-x_k^\sRB$ for $\lambda\neq 0$ and $M=1$}\label{t:ckxkm1}
\centering
\begin{tabular}{|c|c|c|c|c|} \hline
Case & Range of $c_k$ & $\lfloor c_k \rceil_{\cal X}$ 
  & $x_k^\sRB$ & Range of $c_k -x_k^\sRB$  \\ \hline \hline
0 & $ 0 \leq c_k \leq 1/2$ & 
   0    &   0  & 
   $ 0 \leq c_k  \leq 1/2$  \\ \hline
\multirow{2}{*}{ 1} & \multirow{2}{*}{ $1/2 < c_k $} & \multirow{2}{*}{  $1$} & 
         $1$ & $ \alpha_k^{(1)}=\max\{ \alpha_k^{(1)},-1/2\} < c_k - 1 $    \\ \cline{4-5}
    &   &  & 
      $0$ & $ 1/2  < c_k \leq 1+ \alpha_k^{(1)}$  \\ \hline
\end{tabular}
\end{table*}

\begin{theorem} \label{th:babaisp}
Given the regularization parameter $\lambda \in [0,\infty)$, denote \be \label{eq:bardef}
\blam:=\lambda/\sigma^2, \ \ 
\br_{kk}:= r_{kk}/(\sqrt{2}\sigma).
\ee
Let the index $j_k$ be defined by \eqref{eq:j0}, then
\be \label{eq:j0exp}
j_k = \min\left \{M, \left\lfloor \frac{1}{2}\sqrt{1+\frac{\blam}{\br_{kk}^2}} \right\rfloor +1 \right\}.
\ee
The SP of the $L_0$-regularized Babai point $\x^{\sRB}$ for 
the $L_0$-RBILS problem \eqref{eq:l0rbils} satisfies
\be \label{eq:rbsp}
P^\sRB(\R):=\Pr(\x^\sRB=\x^\ast)
=\prod_{k=1}^{n} \rho^{\sRB}_k, 
\ee
where 
\al{
\rho_k^\sRB  
=\ &   \frac{p}{2M}  +  \frac{(M-j_k)p}{M} \erf\left(\br_{kk}\right)   
 -\frac{p}{2M} \erf\left(\alpha_k^{(j_k)} \br_{kk}\right)  
+ (1-p) \erf\left( (2j_k-1+ \alpha_k^{(j_k)}) \br_{kk}\right)  \label{eq:rbksp}
}
with \textup{(}see \eqref{eq:alphakj}\textup{)}
\be  \label{eq:alphakjk}
\alpha_k^{(j_k)}=\frac{\bar\lambda}{2\br_{kk}^2(2j_k-1)}-\frac{1}{2} (2j_k-1);
\ee
and the SP of the unregularized Babai point $\x^\sBB$ for the BILS problem \eqref{eq:bils} satisfies
\be \label{eq:bbsp}
P^\sBB(\R):=\Pr(\x^\sBB=\x^\ast)
=\prod_{i=1}^{n} \rho_k^\sBB, 
\ee
where 
\begin{align}
\rho_k^\sBB
= \frac{p}{2M}  + \frac{(M-1)p}{M} \ \erf\left(\br_{kk} \right)   
 + \left(1- \frac{2M-1}{2M} p \right)\ \erf \left(\frac{1}{2}\br_{kk} \right).
\label{eq:bbksp}
\end{align}
\end{theorem}

\begin{proof}
From \eqref{eq:alphakj} it is easy to verify that the inequality $\alpha_k^{(j)}\geq -1$ is equivalent to
$$
j \leq \frac{1}{2}\sqrt{1 + \frac{\lambda}{2r_{kk}^2}}+1 =  \frac{1}{2}\sqrt{1+\frac{\blam}{\br_{kk}^2}}  +1.
$$

Then by the definition of $j_k$ in \eqref{eq:j0}, we obtain  \eqref{eq:j0exp}.

For $i=1,2,\ldots,n$, denote the events
\be \label{eq:erb}
E_i^\sRB := (x_{i}^\sRB=x^\ast_i, x_{i+1}^\sRB=x^\ast_{i+1}, \ldots, x_n^\sRB = x^\ast_n).
\ee
Then, applying the chain rule of conditional probabilities yields
\beq\label{eq:chain}
\Pr(\x^\sRB=\x^\ast)
 =\Pr(E_1^\sRB) =\prod_{k=1}^{n}\Pr(x_k^\sRB=x^\ast_k|E_{k+1}^\sRB),
\eeq
where $E_{n+1}^\sRB$ is the sample space $\Omega$, so $\Pr(x_n^\sRB=x^\ast_n|E_{n+1}^\sRB)=\Pr(x_n^\sRB=x^\ast_n)$.

In the following we derive a formula for $\Pr(x_k^\sRB=x^\ast_k|E_{k+1}^\sRB)$.
To simplify notation, we write
\be \label{eq:rkrb}
\rho_k^\sRB:=\Pr(x_k^\sRB=x^\ast_k|E_{k+1}^\sRB).
\ee
We consider the two cases that $M\geq 2$ and $M=1$ separately.

We first assume that $M\geq 2$.
Partition $\rho_k^\sRB$ into four parts, corresponding to $\lfloor c_k \rceil_{\cal X}=  
0,  1,  2j-1$ (for $j=2,\ldots,M-1$), and $2M-1$, respectively (see Table \ref{t:ckxk}):
\be \label{eq:xkrb4}
\rho_k^\sRB
= 2 \rho^\sRB_{(0)} +  2\rho^\sRB_{(1)} +\textstyle{ 2\sum_{j=2}^{M-1} \rho^\sRB_{( j)}}+ 2\rho^\sRB_{( M)},
\ee
where the factor 2 is due to the symmetry we mentioned before.

In the following we  derive a formula for each term on the right hand side of \eqref{eq:xkrb4}, which corresponds to a case in Table \ref{t:ckxk}.
In the derivation we consider different possible values of $j_k$
and   use the distributions \eqref{eq:xprob} and \eqref{eq:ckdis}, the probability
formulas \eqref{eq:prxla}--\eqref{eq:praxb} and Lemma \ref{le:akj}.
To simplify notation, we denote 
\be
\gamma:=\br_{kk}=r_{kk}/(\sqrt{2}\sigma).
\label{eq:gamma}
\ee

Case 0 (see Table \ref{t:ckxk}).  In this case,
\aln{
\rho^\sRB_{(0)} 
& = \Pr \big( \big(x^\ast_k= 0,  0  \leq c_k \leq  1/2  \big) \,|\, E_{k+1}^\sRB\big) \\
& =  \Pr( x^\ast_k= 0) \Pr\big(0  \leq c_k \leq  1/2 \,|\, (x^\ast_k= 0, E_{k+1}^\sRB)\big),
}
where in the derivation of the second equality,  we used the following  simple result: 
if $\bE_1$, $\bE_2$ and $\bE_3$ are three events, and $\bE_1$ and $\bE_3$ are independent, then
$$
\Pr((\bE_1,\bE_2)|\bE_3)= \Pr(\bE_1) \Pr(\bE_2 |( \bE_1, \bE_3)) .
$$
Then
\al{
2\rho^\sRB_{(0)}= (1-p)\erf\left(\gamma/2\right).
\label{eq:rrb0}
}

Case 1 (see Table \ref{t:ckxk}). 
In this case, with $\alpha_k^{(1)}\geq -1/2$,
\al{
2\rho^\sRB_{(1)} 
 = \ & 2\Pr( x^\ast_k=1) \Pr\left( \alpha_k^{(1)}<  c_k -1  \leq 1 \,|\, x^\ast_k=1, E_{k+1}^\sRB\right) \notag \\
   & + 2\Pr (x^\ast_k= 0)  
     \Pr\left( 1/2 <  c_k  \leq1+\min\{\alpha_k^{(1)},1\}    \,|\, x^\ast_k=0,E_{k+1}^\sRB\right) 
   \notag \\
= \ & \frac{p}{2M} \biggl \{  \erf\left( \gamma \right)
        - \erf\left( \left[ 1-\big(1- \alpha_k^{(1)} \big)_+ \right] \gamma \right) \notag \\
  & \! + (1-p)\left\{\erf\left( \big[1+\min\{\alpha_k^{(1)},1\} \big] \gamma \right)
   -   \erf\left(\gamma/2\right)\right\}.
   \label{eq:rrb1}
}
If $j_k=1$,  by Lemma \ref{le:akj}, $\alpha_k^{(1)}\leq 1$, then
from \eqref{eq:rrb1} we obtain
\al{
2\rho^\sRB_{(1)} 
 =\ & \frac{p}{2M} \left( \erf\left(\gamma\right) 
            - \erf\left(\alpha_k^{(1)}\gamma\right) \right)  
 + (1-p) \left( \erf\left(\Big(1+\alpha_k^{(1)}\Big)\gamma\right) 
            - \erf\left(\gamma/2\right)\right).
            \label{eq:rrb11}
}           
If $j_k \geq 2$,  by Lemma \ref{le:akj}, $\alpha_k^{(1)} \geq 1$,  then from \eqref{eq:rrb1} we obtain
\aln{
2 \rho^\sRB_{(1)} 
=(1-p) \left( \erf\left(2\gamma\right) 
            - \erf\left(\gamma/2\right)\right).
}

Case $j$, $2\leq j \leq M-1$ (see Table \ref{t:ckxk}). 
To shorten the expressions,  denote the events 
\begin{equation*}
\begin{aligned}
    F_j^1 &:= (\max\{\alpha_k^{(j)},-1\} \leq c_k -(2j-1)  \leq 1),  \\
    F_j^2 &:= ( 2j-2 \leq c_k < 2j-1+ \min\{\alpha_k^{(j)},1\}).
\end{aligned}
\end{equation*}
In this case,
\al{
 2\rho^\sRB_{(j)}  
 =\ & 2\Pr ( 
x^\ast_k=2j-1) \Pr\left(
F_j^1
\,|\, 
    ( x^\ast_k=2j-1, E_{k+1}^\sRB)\right)   
    + 2\Pr( x^\ast_k=0) 
\Pr \left(
F_j^2
\,|\,  
     ( x^\ast_k=0, E_{k+1}^\sRB)\right)  \notag \\
=\ &  \frac{p}{2M} \biggl \{ \erf\left(\gamma \right) 
-   \erf\left(  \Big[ 1- \left(1-\max\{\alpha_k^{(j)},-1\}\right)_+\Big]
      \gamma \right) \biggr\}  \notag \\
  & + (1-p) \Big\{\erf\left( \Big[ 2j-2 + \left(1+\min\{\alpha_k^{(j)},1\}\right)_{+}  \Big]
      \gamma \right)  -  \erf\left( (2j-2) \gamma \right)\Big\}.  
      \label{eq:rrbj}
}      
If $j_k=1$,   by Lemma \ref{le:akj}, $\alpha_k^{(j)} < -1$ for $j=2,\ldots,M$, then from \eqref{eq:rrbj} we obtain
\al{
2\rho^\sRB_{(j)} = \frac{p}{M}  \, \erf\left(\gamma\right), 
\label{eq:rrbj1}
}
and then
\aln{
2\sum_{j=2}^{M-1} \rho^\sRB_{(j)}
= \sum_{j=2}^{M-1}\frac{p}{M}    \, \erf\left(\gamma\right) 
= \frac{M-2}{M}p\,\erf\left(\gamma\right).
}
If $ 2 \leq  j_k \leq M-1$, by Lemma \ref{le:akj},
$\alpha_k^{(j)} \geq 1$ for $j=2,\ldots,j_k-1$,
$-1 \leq \alpha_k^{(j_k)} \leq 1$, and $\alpha_k^{(j)} <-1$ for $j=j_k+1,\ldots,M$, then from \eqref{eq:rrbj} we obtain
\aln{
2\sum_{j=2}^{M-1}  \rho^\sRB_{(j)}
=\ &   \frac{p}{2M}\Big( \erf\left(\gamma\right) 
      -  \erf\left(\alpha_k^{(j_k)}\gamma\right)    
          +  2 (M-j_k-1)  \, \erf\left(\gamma\right)\Big)  \\
   & + (1-p)\Big( \erf\big((2j_k-1+\alpha_k^{(j_k)})\gamma\big)  
            -\erf\left(2\gamma\right) \Big).
}
If $j_k=M$, by Lemma \ref{le:akj}, $\alpha_k^{(j)} \geq 1$ for $j=2,\ldots,M-1$, then from \eqref{eq:rrbj} we obtain
\aln{
2\sum_{j=2}^{M-1} \rho^\sRB_{(j)}
= (1-p) \left( \erf\left( 2(M-1)\gamma\right) 
           -\erf\left(2\gamma\right) \right).
}

Case $M$. 
To shorten the expressions, denote the events 
\begin{equation*}
\begin{aligned}
    F_M^1 &:= (\max\{\alpha_k^{(M)},-1\} < c_k-(2M-1)), \\
    F_M^2 &:= (2M -2 < c_k \leq 2M -1 +\alpha_k^{(M)}).
\end{aligned}
\end{equation*}
In this case,
\al{
  2\rho^\sRB_{(M)}  
=\ & 2\Pr(x^\ast_k= 2M-1) \Pr\big(F_M^1    \,|\, (x^\ast_k= 2M-1, E_{k+1}^\sRB)\big)   
+ 2\Pr \big(x^\ast_k= 0\big) \Pr\big( F_M^2  \,|\, E_{k+1}^\sRB\big)  \notag \\
=\ & \frac{p}{2M} \left\{1 - \erf\left(\max\{\alpha_k^{(M)},-1\} \gamma\right)\right\}  \notag \\
& + (1-p)\Big\{ \erf\left(\left[ 2M-2+ (\alpha_k^{(M)}+1)_+ )  \right] \gamma \right) -  \erf\left(\left[2M-2\right] \gamma \right) \Big\}.
 \label{eq:rrbM}
}
If $1\leq j_k \leq M-1$, by Lemma \ref{le:akj}, $\alpha_k^{(M)} < -1$, then
\al{
2\rho^\sRB_{(M)} = \frac{p}{2M} \left(1 + \erf\left(\gamma\right)\right).
\label{eq:rrbM1}
}
If $j_k=M$, by Lemma \ref{le:akj}, $-1 \leq \alpha_k^{(M)} \leq 1$, then
\aln{
2\rho^\sRB_{(M)} 
= \ & \frac{p}{2M}\left(1 - \erf\left(\alpha_k^{(M)}\gamma\right)\right) 
+(1-p) \Big\{ \erf\left(\Big(2M-1+\alpha_k^{(M)}\Big)\gamma\right) 
- \erf\left((2M-2)\gamma\right)\Big\}.
 }

So far for $M\geq 2$ we have obtained $2 \rho^\sRB_{(0)}$,  $2\rho^\sRB_{(1)}$, $2\sum_{j=2}^{M-1} \rho^\sRB_{( j)}$ and $2\rho^\sRB_{(M)}$,
for each possible value of $j_k$.
To find $\rho_k^\sRB$ in \eqref{eq:xkrb4},
we add all these terms together for $j_k=1$, $2\leq j_k \leq M-1$ and $j_k=M$, respectively.
It is easy to verify that we have a unified expression for different values of $j_k$:  
\aln{
 \rho_k^\sRB 
=    \frac{p}{2M}  + (M-j_k) \frac{p}{M} \erf\left(\gamma\right)   
  - \frac{p}{2M} \erf\left(\alpha_k^{(j_k)} \gamma\right)    
  + (1-p) \erf\left(\left[ 2j_k-1 + \alpha_k^{(j_k)}\right]\gamma\right),
}
leading to \eqref{eq:rbksp}.


We now show that \eqref{eq:rbksp} also holds when $M=1$
(corresponding to Table \ref{t:ckxkm1}).
Obviously from \eqref{eq:j0exp} we see that $j_k=1$.
We partition $\rho_k^\sRB$ into two parts, corresponding to $\lfloor c_k \rceil_{\cal X}=  
0,  1$, respectively (see the two cases in Table \ref{t:ckxkm1}):
\be \label{eq:xkrbm1}
\rho_k^\sRB
= 2 \rho^\sRB_{(0)} +  2\rho^\sRB_{(1)},
\ee
where $2\rho^\sRB_{(0)}$ is identical to that in \eqref{eq:rrb0}. 
The second term on the right hand side of \eqref{eq:xkrbm1} can be derived as follows:
\al{
 2\rho^\sRB_{(1)}  
 = \ &  2\Pr( x^\ast_k=1) 
 \Pr\big( \alpha_k^{(1)}<  c_k -1  \,|\, (x^\ast_k=1, E_{k+1}^\sRB)\big) \notag \\
   & + 2\Pr (x^\ast_k= 0)  \Pr\big(1/2 < c_k  \leq 1+ \alpha_k^{(1)}  \,|\, (x^\ast_k=0,E_{k+1}^\sRB)\big) 
   \notag \\
 = \  & \frac{p}{2M} \left \{ 1 - \erf\left(\alpha_k^{(1)} \gamma \right) \right\} 
  + (1-p)\left\{\erf\left( \big[1+\alpha_k^{(1)} \big] \gamma \right)  -   \erf\left(\gamma/2\right)\right\}.
   \label{eq:rrbm11}
}
Then from \eqref{eq:xkrbm1}, \eqref{eq:rrb0} and \eqref{eq:rrbm11} we obtain
\aln{
\rho_k^\sRB
=\ & (1-p) \erf\left(\gamma/2\right)
+ \frac{p}{2M} \left \{ 1 - \erf\left(\alpha_k^{(1)} \gamma \right) \right\}  + (1-p)\left\{\erf\left( \big[1 + \alpha_k^{(1)} \big] \gamma \right)  -   \erf\left(\gamma/2\right)\right\} \\
=\ & \frac{p}{2M} \left \{ 1 - \erf\left(\alpha_k^{(1)} \gamma \right) \right\} 
   + (1-p)\ \erf\left( \big[1+\alpha_k^{(1)} \big] \gamma \right),
  }
which is identical to  the expression given in \eqref{eq:rbksp},
where  $j_k=1$, i.e., \eqref{eq:rbksp} still holds when $M=1$.

The remaining part of the proof is to show \eqref{eq:bbsp}-\eqref{eq:bbksp}.
Recall that when the regularization parameter $\lambda=0$,
$\x^\sRB$ becomes $\x^\sBB$.
Thus, if we take $\lambda=0$, \eqref{eq:rbsp} just becomes \eqref{eq:bbsp}.
What we need to verify is that \eqref{eq:rbksp} becomes   \eqref{eq:bbksp}.
In fact, when $\lambda=0$, from \eqref{eq:j0exp} and \eqref{eq:alphakjk} 
we obtain
$$
j_k=1, \ \ 
\alpha_k^{(1)} = -\frac{1}{2}.
$$
Then \eqref{eq:rbksp} becomes
\aln{
\rho_k^\sBB
  =\ &  \frac{p}{2M} + \frac{(M-1)p}{M} \erf(\br_{kk})
  - \frac{p}{2M} \erf\Big(-\frac{1}{2}\br_{kk}\Big)  
    + (1-p)\erf\Big(\frac{1}{2}\br_{kk}\Big).  
}
Thus, \eqref{eq:bbksp} holds, completing the proof. 
\end{proof}

Here we make some remarks about Theorem \ref{th:babaisp} and its proof.


\begin{remark}
In \cite{ZhuG11} a formula for the symbol error rate (SER) $1-\Pr(x_k^\sRB=x^\ast_k|E_{k+1}^\sRB)$ is derived for
the special case $M=1$. 
It is stated in \cite{ZhuG11} that for a general $2M$-ary constellation the SER can be approximated  using the union bound. 
Here our formula for $\Pr(x_k^\sRB=x^\ast_k|E_{k+1}^\sRB)$ is a rigorous one and it leads
to the formula for $\Pr(\x^\sRB=\x^\ast)$. 
\end{remark}

\begin{remark}
In \cite{WenC17} a formula for $\Pr(\x^\sBB=\x^\ast)$  is derived for the case that 
$\x^\ast$ is uniformly distributed over a box ${\cal X}^n$ which includes all consecutive integer points,
so \eqref{eq:bbsp} cannot be obtained from that formula.
\end{remark}

\begin{remark}
An earlier version of Theorem \ref{th:babaisp}
was presented in the conference paper \cite{ChaP18}. 
Here the expression of $\Pr(\x^\sRB=\x^*)$ is much more concise 
(although the derivation is longer), 
and it facilitates the development of new theoretical results and 
column permutation strategies to be given later.
\end{remark}

\begin{corollary} \label{cor:rbarlimit}
Let the regularization  parameter $\lambda$ be any positive number.
We have
\aln{
& \lim_{r_{kk}\rightarrow 0} \rho_k^\sRB = 1-p, \ \
\lim_{r_{kk}\rightarrow \infty} \rho_k^\sRB = 1, \\
& \lim_{r_{kk}\rightarrow 0} \rho_k^\sBB = \frac{p}{2M}, \ \
\lim_{r_{kk}\rightarrow \infty} \rho_k^\sBB = 1.
}
\end{corollary}


\begin{proof}  
When $r_{kk}$ is small enough, from \eqref{eq:j0exp} we see $j_k=M$.
Then from \eqref{eq:alphakjk},
$$
\lim_{r_{kk}\rightarrow 0} \alpha_k^{(j_k)} = \infty.
$$
Thus, from \eqref{eq:rbksp} we obtain
$$
\lim_{r_{kk}\rightarrow 0} \rho_k^\sRB
= \frac{p}{2M} - \frac{p}{2M}\cdot 1  + (1-p)\cdot 1 = 1-p.
$$

When $r_{kk}$ is large enough, from \eqref{eq:j0exp} we see $j_k=1$.
From \eqref{eq:alphakjk},   
$\lim_{r_{kk}\rightarrow \infty} \alpha_k^{(j_k)} = -1/2$, 
and then from \eqref{eq:rbksp} we obtain
$$
\lim_{r_{kk}\rightarrow \infty} \rho_k^\sRB
= \frac{p}{2M} + \frac{(M-1)p}{M} + \frac{p}{2M} + (1-p) = 1.
$$

From  \eqref{eq:bbksp} we obtain
$$
\lim_{r_{kk}\rightarrow 0} \rho_k^\sBB
= \frac{p}{2M},
$$
$$
\lim_{r_{kk}\rightarrow \infty} \rho_k^\sBB
= \frac{p}{2M} + \frac{(M-1)p}{M} + (1-\frac{2M-1}{2M} p) =1.
$$
\end{proof}

We make a remark about this corollary.
\begin{remark}
Under the condition \eqref{eq:pinequ}, from Corollary \ref{cor:rbarlimit}
we observe 
$$
\lim_{r_{kk}\rightarrow 0} \rho_k^\sRB 
\geq \lim_{r_{kk}\rightarrow 0} \rho_k^\sBB.
$$
If \eqref{eq:pinequ} is a strict inequality, then 
when $r_{kk}$ is small enough, we have
$$
\rho_k^\sRB >  \rho_k^\sBB.
$$
For a general $r_{kk}$, this inequality may not hold.
But it still holds when the regularization parameter $\lambda$
satisfies \eqref{eq:lambda}, see Theorem \ref{th:relation}.
\end{remark}

\section{Some properties of the $L_0$-regularized Babai point.}\label{sec:bproperties}
In this section we present some interesting properties of the $L_0$-regularized Babai point $\x^\sRB$.
First, we show that 
the SP of $\x^\sRB$ is higher than
the SP of $\x^\sBB$ when $\lambda$ in  \eqref{eq:l0rbils} is defined as $\lambda^*$ in \eqref{eq:lambda},
indicating the \textit{a priori} information about the distribution of the true parameter vector is useful even for a sub-optimal solution to the MAP detection.
Then we show that
the SP of $\x^\sRB$ 
is an increasing function
of the diagonal elements of the matrix $\R$ in \eqref{eq:qr}. 
This property allows us to develop an efficient column permutation strategy to increase the SP of the 
$L_0$-regularized Babai point in Section \ref{sec:spgalg}.
Finally, we present a bound on the SP of $\x^\sRB$,
which gives us motivation to investigate the effect of various column permutation strategies in Section \ref{sec:permut}.

\subsection{Some basic properties of $\rho_k^\sRB$}
In this subsection, we give some results on the continuity of the function $\rho_k^\sRB$, which will be used in some later proofs. 

\begin{lemma}\label{le:rkcts}
Let the regularization parameter $\lambda=\lambda^*$ (see \eqref{eq:lambda}).
With the same notation as in Theorem \ref{th:babaisp}, 
\al{
&  \frac{\partial\rho_k^\sRB}{\partial \br_{kk}} = \frac{2}{\sqrt{\pi}} \Bigg[ \frac{M-j_k}{M}pe^{-\br_{kk}^2}  
    + (1-p)(2j_k-1)e^{-\left(\frac{1}{2(2j_k-1)}\frac{\bar\lambda}{\br_{kk}}+\frac{1}{2}(2j_k-1)\br_{kk}\right)^2}  \Bigg], \label{eq:drk}
}
and it is continuous.
\end{lemma}

\begin{proof}
To simplify notation, we denote $\gamma = \br_{kk}$ as in \eqref{eq:gamma}, and $\blam=\lambda^*/\sigma^2$.
To emphasize that $\rho_k^\sRB$ is a function of $\gamma$, we rewrite it as $\rho_k^\sRB(\gamma)$.
In order to prove our claim, we need only to show that $\rho_k^\sRB(\gamma)$ and its first derivative with respect to $\gamma$ are continuous for $\gamma > 0$.
From \eqref{eq:j0exp}, when $\gamma$ is increasing, $j_k$ either remains the same or is decreasing.
Let $\gamma^{(i)}$ satisfy 
\be \label{eq:gj}
M-i+1=  \frac{1}{2} \sqrt{1+\frac{\bar\lambda}{(\gamma^{(i)})^2}}  +1, \ \ i=1,\ldots,M-1.
\ee
Denote 
\be \label{eq:g0m}
\gamma^{(0)}=0, \ \ \gamma^{(M)}= \infty.
\ee
Then, from  \eqref{eq:j0exp},
\be \label{eq:jkpw}
\gamma \in (\gamma^{(i-1)}, \gamma^{(i)}] 
\Rightarrow j_k = M-i+1, \ \ i=1,\ldots,M.
\ee
Note that if $M=1$, then \eqref{eq:gj} does not exist,
but \eqref{eq:jkpw} still holds with 
$(\gamma^{(0)},\gamma^{(1)})=(0,\infty)$ (here the 
right parenthesis instead of the right bracket has to be used).
Here we consider a general $M$, but the proof can be applied
to the special case that $M=1$ by some simplification.

From \eqref{eq:gj} it follows that for $i=1,\ldots,M-1$,
\begin{equation}\label{eq:giexp}
    \gamma^{(i)} = \sqrt{\frac{\blam}{(2(M-i+1)-1)(2(M-i+1)-3)}},
\end{equation}
which will be used later. 

From \eqref{eq:rbksp} and \eqref{eq:alphakjk}, it is easy to see that 
for any $i=1,\ldots,M$, when $\gamma \in (\gamma^{(i-1)}, \gamma^{(i)})$ 
the function $\rho^\sRB_k$ is differentiable, and we have
\al{
 \rho_k^\sRB(\gamma)  
  = \ &    \frac{p}{2M} +  \frac{M-j_k}{M} p \, \erf\left(\gamma\right) 
  - \frac{p}{2M} \, \erf\left(\frac{1}{2(2j_k-1)}\frac{\bar\lambda}{\gamma}-\frac{1}{2}(2j_k-1)
 \gamma\right)  \notag  \\
&   + (1-p)\, \erf\left(\frac{1}{2(2j_k-1)}\frac{\bar\lambda}{\gamma}+\frac{1}{2}(2j_k-1)\gamma\right),   \label{eq:rkg}  \\
&  \frac{\partial\rho_k^\sRB(\gamma)}{\partial \gamma}= \frac{2}{\sqrt{\pi}} \Big[ \frac{M-j_k}{M}pe^{-\gamma^2}   
   + (1-p)(2j_k-1)e^{-\left(\frac{1}{2(2j_k-1)}\frac{\bar\lambda}{\gamma}+\frac{1}{2}(2j_k-1)\gamma\right)^2}  \Big]. \label{eq:drkg}
}
Furthermore, it is easy to see that 
$\partial\rho_k^\sRB(\gamma) / \partial \gamma$ is continuous 
when $\gamma \in (\gamma^{(i-1)}, \gamma^{(i)})$. 
Therefore to prove our claim, 
it suffices to show that $\rho^\sRB_k(\gamma)$ 
and $\partial\rho_k^\sRB(\gamma) / \partial \gamma$ 
are continuous at $\gamma^{(i)}$ for $i=1,\ldots,M-1$, where $j_k$ changes,
i.e.,  
we need only to show that 
 for $i=1,\ldots,M-1$,
\al{
& \rho_k^\sRB(\gamma^{(i)}) = \lim_{\gamma \to \gamma^{(i)-}}  \rho_k^\sRB(\gamma)
 = \lim_{\gamma \to \gamma^{(i)+}}  \rho_k^\sRB(\gamma), \label{eq:rcont} \\
& \frac{\partial \rho_k^\sRB(\gamma^{(i)})}{\partial \gamma} 
= \lim_{\gamma \to \gamma^{(i)-}}  \frac{\partial\rho_k^\sRB(\gamma)}{\partial \gamma}
 = \lim_{\gamma \to \gamma^{(i)+}}  \frac{\partial\rho_k^\sRB(\gamma)}{\partial \gamma}. \label{eq:drcont}
} 

The first equality of \eqref{eq:rcont} trivially holds by \eqref{eq:jkpw}.
From \eqref{eq:rkg}, \eqref{eq:jkpw} and \eqref{eq:giexp}, we can easily verify that the second equality of \eqref{eq:rcont}
can be rewritten as
\aln{
 &  \frac{p}{2M} + \frac{i-1}{M}p\, \erf\big(\gamma^{(i)}\big)   
  + \frac{p}{2M} \,  \erf \big( \gamma^{(i)} \big)    
   + (1-p)\,  \erf\big( 2(M-i)\gamma^{(i)} \big)  \\
= \ &
 \frac{p}{2M} +  \frac{i}{M}p\, \erf\big(\gamma^{(i)}\big)  
  - \frac{p}{2M} \, \erf\big(\gamma^{(i)}\big)  
    +(1-p)\,  \erf\big(2(M-i)\gamma^{(i)}\big),
}
which obviously holds.

Now, to show that the second equality of 
\eqref{eq:drcont} holds, by \eqref{eq:drkg} we need to show
\aln{
& \frac{i-1}{M}pe^{-\left(\gamma^{(i)}\right)^2}  + (1-p)(2M-2i+1)e^{\left( 2(M-i)\gamma^{(i)} \right)^2} \\
=\ & \frac{i}{M}pe^{-\left(\gamma^{(i)}\right)^2} + (1-p)(2M-2i-1)e^{\left( 2(M-i)\gamma^{(i)} \right)^2},
}
or equivalently,
\begin{equation*}
    e^{-\left(\gamma^{(i)}\right)^2} = \frac{1-p}{p/(2M)} e^{-\left( 2(M-i)\gamma^{(i)} \right)^2}.
\end{equation*}
Taking logarithm on both sides and substituting $\blam = \ln\frac{1-p}{p/(2M)}$ (see \eqref{eq:lambda}), the above equation becomes
\aln{
-\frac{\blam}{(2(M-i+1)-1)(2(M-i+1)-3)}  
=  \blam - \frac{4(M-i)^2\blam}{(2(M-i+1)-1)(2(M-i+1)-3)}
}
which holds for all $i = 1,\ldots,M-1$. Therefore the second equality of \eqref{eq:drcont} holds.

To show the first equality of \eqref{eq:drcont}, we use the continuity of $\rho_k^\sRB(\gamma)$.
In fact,
\aln{
\lim_{\gamma \to \gamma^{(i)}} \frac{\partial\rho_k^\sRB(\gamma)}{\partial\gamma}
& = \lim_{\gamma \to \gamma^{(i)}} 
\lim_{\epsilon \to 0} \frac{\rho^\sRB_k (\gamma + \epsilon) - \rho^\sRB_k (\gamma) }{\epsilon}  
 = \lim_{\epsilon \to 0} 
\frac{\rho^\sRB_k (\gamma^{(i)} + \epsilon) - \rho^\sRB_k (\gamma^{(i)}) }{\epsilon}  
  = \frac{\partial \rho_k^\sRB(\gamma^{(i)})}{\partial \gamma},
}
completing the proof.
\end{proof}

\begin{lemma}\label{le:rlpcts}
Let $p$ satisfy \eqref{eq:zeropr}, and the regularization parameter $\lambda$ be positive. Then:
\begin{enumerate}
    \item $\rho^\sRB_k$ is continuous with respect to $\lambda$ for any fixed $p$. 
    \item $\rho^\sRB_k$ is continuous with respect to $p$ when $\lambda$ is defined as in \eqref{eq:lambda}. 
\end{enumerate}
\end{lemma}

\begin{proof}
To simplify notation, We denote $\gamma=\br_{kk}$ as in \eqref{eq:gamma}, and $\blam = \lambda/\sigma^2$. 

To prove our first claim, it suffices to show that $\rho^\sRB_k$ is continuous with respect to $\blam$. 
From \eqref{eq:j0exp} we  see that $j_k$  either remains the same or increases when $\blam$ increases in $[0,\infty)$. 
Let $\blam^{(i)}$ be such that
\al{
i =  \frac{1}{2} \sqrt{1+\frac{{\bar\lambda}^{(i)}}{ \gamma^2}},\quad i=1,\ldots,M-1.
\label{eq:lbi1}
}
Denote 
$$
\blam^{(0)} = 0, \quad \blam^{(M)} = \infty.
$$
From  \eqref{eq:j0exp} we obtain
\be \label{eq:jkpw1}
\blam\in [\blam^{(i-1)},\blam^{(i)}) \Rightarrow j_k = i,  \qquad i=1,\ldots,M.
\ee
Note that if $M=1$, then \eqref{eq:lbi1} does not exist, but \eqref{eq:jkpw1} still holds
with $[\blam^{(0)}, \blam^{(1)})=[0,\infty)$.
Our following proof is for a general $M$, but it is easy to see that 
it can be applied to the special case that $M=1$ by some simplification.

We first show that $\rho_k^\sRB$ is continuous with respect to $\blam$. From \eqref{eq:lbi1},
\al{
\blam^{(i)} = \gamma^2 (2i-1)(2i+1), \ \ i=1,\ldots,M-1.
\label{eq:lbi2} 
}
From \eqref{eq:rbksp} and \eqref{eq:alphakjk}, 
it is obvious that $\rho_k^\sRB$ is continuous 
when $\blam \in (\blam^{(i-1)}, \blam^{(i)})$ for any $i=1,\ldots,M$.

To show the continuity of $\rho_k^\sRB$ at $\blam^{(i)}$ for $i=1,\ldots,M-1$, we need to
show that for any fixed $\gamma$, 
\al{
\lim_{\blam\to \blam^{(i)-}}\rho_k^\sRB  = \lim_{\blam\to \blam^{(i)+}}\rho_k^\sRB = \rho_k^\sRB |_{\blam = \blam^{(i)}}. \label{eq:lcont}
}
The second equality of \eqref{eq:lcont} trivially holds by \eqref{eq:jkpw1}. To show the first equality of \eqref{eq:lcont}, 
from \eqref{eq:rbksp} and \eqref{eq:jkpw1} 
we need to show
\al{
& \frac{p}{M}\erf(\gamma) - \frac{p}{2M}\erf\left( \frac{\blam^{(i)}}{2(2i-1)\gamma} - \frac{1}{2}(2i-1)\gamma \right)  
  + (1-p)\erf\left( \frac{\blam^{(i)}}{2(2i-1)\gamma} + \frac{1}{2}(2i-1)\gamma \right) \notag \\
=& - \frac{p}{2M}\erf\left( \frac{\blam^{(i)}}{2(2i+1)\gamma} - \frac{1}{2}(2i+1)\gamma \right)  
 + (1-p)\erf\left( \frac{\blam^{(i)}}{2(2i+1)\gamma} + \frac{1}{2}(2i+1)\gamma \right). \label{eq:lctseq1}
}
After substituting \eqref{eq:lbi2} into \eqref{eq:lctseq1}, the equality becomes
\aln{
& \frac{p}{M}\erf(\gamma) - \frac{p}{2M}\erf\left(\frac{1}{2}(2i+1)\gamma - \frac{1}{2}(2i-1)\gamma\right)   
  + (1-p)\erf\left( \frac{1}{2}(2i+1)\gamma + \frac{1}{2}(2i-1)\gamma \right)  \\
=& -\frac{p}{2M}\erf\left( \frac{1}{2}(2i-1)\gamma - \frac{1}{2}(2i+1)\gamma \right)   
 + (1-p)\erf\left( \frac{1}{2}(2i-1)\gamma + \frac{1}{2}(2i+1)\gamma \right),
}
which can easily be verified to hold. 
This concludes the proof for the first claim.

To prove our second claim, we first see from \eqref{eq:j0exp}
that $j_k$  either remains the same or is decreasing when $p$ 
is increasing from $0^+$ to $2M/(2M+1)$.
Let $\blam^{(i)}$ be defined as in \eqref{eq:lbi1}, and $p^{(i)}$ be such that
\begin{equation}\label{eq:piexp}
    \begin{aligned}
        & \blam^{(i)}:=\ln \frac{1-p^{(i)}}{p^{(i)}/(2M)}, \ \ i=1,\ldots,M-1.  \\
        & p^{(0)}:=2M/(2M+1),\ \ p^{(M)}:=0, \\ 
& \blam^{(0)}:=\ln \frac{1-p^{(M)}}{p^{(M)}/(2M)}=0, \\
& \blam^{(M)}:=\lim_{p \rightarrow {p^{(0)}}^{+}}\ln \frac{1-p}{p/(2M)}=\infty.
    \end{aligned}
\end{equation}
When  $p \in (p^{(i)}, p^{(i-1)}]$, equivalently we have
$\blam \in [\lambda^{(i-1)},\lambda^{(i)})$.
Using the same approach for proving our first claim, we can obtain the continuity of $\rho^\sRB_k$ with respect to $p$. This concludes our proof.
\end{proof}

\subsection{The superiority of $\x^\sRB$}
In the following we show that $\lambda^*$ is optimal in maximizing
the SP of the $L_0$-regularized Babai point $\x^\sRB$.

\begin{theorem} \label{th:relation}
Let $\lambda\in [0,\infty)$ be an arbitrary regularization parameter
used in $\x^\sRB$
and let $\lambda^*$ be defined in \eqref{eq:lambda}. 
When $\lambda=\lambda^*$, the SP of the $L_0$-regularized Babai point $\x^\sRB$ reaches the global maximum with respect to $\lambda$. 
In particular, when $\lambda=\lambda^*$, for the box-constrained Babai point $\x^\sBB$, 
\be \label{eq:probineq}
\Pr(\x^\sRB=\x^*) \geq \Pr(\x^\sBB=\x^*),
\ee
with equality if and only $\lambda^*=0$, i.e., $p=2M/(2M+1)$.
\end{theorem}

\begin{proof}
In the following we show that for any fixed $\br_{kk}$, $\rho_k^\sRB$ reaches the global maximum with respect to $\lambda$ when $\lambda = \lambda^*$.
This is the key part of the whole proof.

In Lemma \ref{le:rlpcts} we have shown that $\rho^\sRB_k$ is continuous with respect to $\lambda$. 
From \eqref{eq:jkpw1} we observe that $j_k$ does not change 
when $\blam\in[\blam^{(i-1)}, \blam^{(i)})$ for any $i=1,\ldots,M$. 
Therefore, by \eqref{eq:rbksp} $\rho_k^\sRB$ is differentiable 
with respect to $\blam\in (\blam^{(i-1)}, \blam^{(i)})$. 
To simplify notation, let $\eta := (2j_k - 1)\gamma$. 
When $\blam\in[\blam^{(i-1)}, \blam^{(i)})$, from \eqref{eq:rbksp} 
and \eqref{eq:jkpw1} we have  
\al{
\frac{\partial\rho_k^\sRB}{\partial\blam} 
= \frac{1}{\sqrt{\pi}\eta}\Big[(1-p) e^{-\left( \frac{\blam}{2\eta} + \frac{1}{2}\eta \right)^2} 
- \frac{p}{2M}e^{-\left(\frac{\blam}{2\eta} - \frac{1}{2}\eta\right)^2}\Big],
\label{eq:rbkd}
}
which is a right derivative if $\blam=\blam^{(i-1)}$. 
Setting $\frac{\partial\rho_k^\sRB}{\partial\blam}=0$,
we obtain
\al{
(1-p) e^{-\left( \frac{\blam}{2\eta} + \frac{1}{2}\eta \right)^2} 
- \frac{p}{2M}e^{-\left(\frac{\blam}{2\eta} - \frac{1}{2}\eta\right)^2} = 0.
\label{eq:lstationary1}
}
leading to
\al{
\blam= \ln \frac{1-p}{p/(2M)}.
\label{eq:optimalblam}
}
Note that $\lambda^* =\sigma^2 \ln \frac{1-p}{p/(2M)}$ 
(see \eqref{eq:lambda}).
Thus, from \eqref{eq:optimalblam}, the stationary point of $\rho_k^\sRB$ with respect to $\blam$
is $\blam^*:=\lambda^*/\sigma^2$.

Now we show that $\blam^*$ is the maximizer of $\rho_k^\sRB$.
Assume that $\blam^\ast \in [\blam^{(i^\ast-1)}, \blam^{(i^\ast)})$ 
for some $i^\ast \in \{1, \ldots, M\}$. 
For any $\blam \in [\blam^{(i-1)}, \blam^{(i)})$ with $i < i^\ast$, 
we have $j_k = i$ and $e^\blam < \frac{1-p}{p/(2M)}$. 
Therefore
\aln{
\frac{1-p}{p/(2M)}e^{-\left( \frac{\blam}{2\eta} + \frac{1}{2}\eta \right)^2} - e^\blam e^{-\left( \frac{\blam}{2\eta} + \frac{1}{2}\eta \right)^2} > 0.
}
Then from \eqref{eq:rbkd} we observe that 
$\frac{\partial\rho_k^\sRB}{\partial\blam} > 0$. 
Using a symmetrical argument, we can show that for any $\blam \in [\blam^{(i-1)}, \blam^{(i)})$ with $i > i^\ast$, $\frac{\partial\rho_k^\sRB}{\partial\blam} < 0$.
We can also use the same argument to show that 
when $\blam\in[\blam^{(i^\ast-1)}, \blam^{(i^\ast)})$,
$\frac{\partial\rho_k^\sRB}{\partial\blam} > 0$ when $\blam < \blam^\ast$, 
and $\frac{\partial\rho_k^\sRB}{\partial\blam} < 0$ when $\blam > \blam^\ast$.
With the above results and the continuity of $\rho_k^\sRB$ with respect to $\blam$ from Lemma \ref{le:rlpcts}, 
we can conclude that $\rho_k^\sRB$ is strictly  increasing
when $\blam < \blam^\ast$, 
and strictly  decreasing when $\blam > \blam^\ast$. 
Thus, $\blam^\ast$ is the only 
global maximizer of $\rho_k^\sRB$ with respect to $\blam$, 
or equivalently  $\lambda^\ast$ is the only 
global maximum of $\rho_k^\sRB$ with respect to $\lambda$. 
This concludes the proof for our claim.

Using the result above, from \eqref{eq:rbsp} we see
it is obvious that the SP of $\x^\sRB$
reaches the global maximum  with respect to $\lambda$ if and only if $\lambda=\lambda^*$.

Recall that when $\lambda=0$, $\x^\sRB$ becomes $\x^\sBB$.
Thus the inequality \eqref{eq:probineq} must hold
and it becomes an equality if and only if $\lambda^*=0$,
i.e., $p=2M/(2M+1)$.
\end{proof}

We make some remarks about Theorem \ref{th:relation}.
\begin{remark}
The inequality \eqref{eq:probineq} in Theorem \ref{th:relation} was presented in the conference paper \cite{ChaP18}, but its proof there was completely different 
from the proof given here, as the global optimality result 
of the regularization parameter $\lambda$
proved in 
Theorem \ref{th:relation} 
is new.
\end{remark}

\begin{remark}
    Theorem \ref{th:relation} shows the advantage of applying $L_0$ regularization, but it does not show the degree of improvement the regularization to the SP. 
    Later in Section \ref{sec:numexp_the_exp_sp}, Figure \ref{fig:spplots} shows that the $L_0$ regularization can give significant SP improvement, especially when $p$ is small, which leads to $\x^\ast$ being sparse. 
\end{remark}

\begin{remark}
If $\lambda$ is an arbitrary positive number, i.e., it is not defined by \eqref{eq:lambda},
then the inequality \eqref{eq:probineq} may not hold.
For example, for 
$$
\R = \bmxc{cr}
0.8432 & -0.6045 \\
0 & 0.8980 \\
\emxc, \ \
p = 0.6, \ \ M = 4, \ \ \sigma = 0.2,
$$
we have
$\Pr(\x^\sBB = \x^\ast) = 0.9718$. 
With $\lambda = 0.0670$ calculated by \eqref{eq:lambda}, 
$\Pr(\x^\sRB = \x^\ast) = 0.9805$,
while with $\lambda = 0.2, \Pr(\x^\sRB = \x^\ast) = 0.9537$.
\end{remark}

\subsection{Monotonicity properties}
The following result shows the monotonicity property of 
$P^\sRB(\R)$ and $P^\sBB(\R)$ with respect to the diagonal elements of $\R$.

\begin{theorem} \label{th:rmono}
Let the regularization parameter $\lambda$ used by $\x^\sRB$ be chosen as $\lambda^*$ defined by  \eqref{eq:lambda}.
With the same notation as in Theorem \ref{th:babaisp}, 
for each $k=1,2,\ldots,n$, $\rho_k^\sRB$ and $\rho_k^\sBB$ are strictly increasing for $r_{kk} \in (0,\infty)$.
Consequently, both $P^\sRB(\R)$ and $P^\sBB(\R)$ 
are strictly increasing over $r_{kk} \in (0,\infty)$ for $k=1,2,\ldots,n$.
\end{theorem}

\begin{proof}
It is easy to see from \eqref{eq:bbksp} that $\rho_k^\sBB$ is strictly increasing for $\br_{kk} \in (0,\infty)$, or for $r_{kk} \in (0,\infty)$.

Now we consider $\rho_k^\sRB$.
If $\lambda^*=0$, it becomes $\rho_k^\sBB$.
Thus, in the following  we assume that $\lambda^*>0$. 
As $1\leq j_k \leq M$, 
it is easy to see from \eqref{eq:drk} in Lemma \ref{le:rkcts} that the derivative $\partial \rho_k^\sRB / \partial \br_{kk} > 0$ for any $\br_{kk} > 0$. Therefore $\rho_k^\sRB$ is strictly increasing for $\br_{kk} \in (0,\infty)$ or for $r_{kk} \in (0,\infty)$. 

The second conclusion is obvious from \eqref{eq:rbsp} and \eqref{eq:bbsp} in Theorem \ref{th:babaisp}.
\end{proof}

\begin{remark}
If $\lambda$ is an arbitrary positive number, $\rho_k^\sRB$ may not be increasing 
for $\br_{kk} \in (0,\infty)$, 
and as a result the monotonicity property of $P^\sRB(\R)$ does not hold.
It is easy to find an example to illustrate it.
\end{remark}

Theorem \ref{th:relation} shows that when $\lambda^*$ is used as the regularization parameter,
the SP of $\x^\sRB$ is higher than the SP of $\x^\sBB$,
and the two are equal when $p=2M/(2M+1)$.
The following theorem shows how their ratio changes 
when $p$ in \eqref{eq:zeropr} changes from $0^+$ to $2M/(2M+1)$.

\begin{theorem}\label{th:spratiomono}
Let the regularization parameter $\lambda$ used by $\x^\sRB$ be chosen as $\lambda^*$ defined by  \eqref{eq:lambda}.
The ratio $\Pr(\x^\sRB=\x^\ast)/\Pr(\x^\sBB=\x^\ast)$ is strictly decreasing for $p \in (0, 2M/(2M+1))$. 
\end{theorem}

\begin{proof}
From Theorem \ref{th:babaisp}, we have
\begin{equation*}
\frac{\Pr(\x^\sRB=\x^\ast)}{\Pr(\x^\sBB=\x^\ast)}
= \prod_{k=1}^n \frac{\rho^\sRB_k}{\rho^\sBB_k}.
\end{equation*}
Therefore it is sufficient to prove that for each $k=1,2,\ldots,n$, the ratio
$\rho^\sRB_k/\rho^\sBB_k$ is strictly decreasing for $p \in (0, 2M/(2M+1))$. 
Since 
$\rho^\sRB_k/\rho^\sBB_k
=(\rho^\sRB_k-\rho^\sBB_k)/\rho^\sBB_k+1$,
it suffices to show that $\rho^\sRB_k-\rho^\sBB_k$ is 
strictly decreasing and  $\rho^\sBB_k$ is increasing 
for $p \in (0, 2M/(2M+1))$.

To simplify notation, we use $\gamma$ to replace $\br_{kk}$
in the expressions for $\rho^\sRB_k$ and  $\rho^\sBB_k$.
From \eqref{eq:bbksp}, $\rho^\sBB_k$ is continuous with respect to $p$, and we obtain
\al{
\frac{\partial \rho^\sBB_k}{\partial p}
&= \frac{1}{2M} + \frac{M-1}{M} \erf(\gamma) - \frac{2M-1}{2M}\erf\left(\frac{1}{2}\gamma\right) \label{eq:dbbsp} \\
& = \frac{1}{2M}(1-\erf(\gamma)) + \frac{2M-1}{2M}\left(\erf(\gamma)-\erf\left(\frac{1}{2}\gamma\right)\right) \notag \\
& >0. \notag
}
Thus $\rho^\sBB_k$ is strictly increasing for $p \in (0, 2M/(2M+1))$.

The remaining part is to show that 
$\rho^\sRB_k-\rho^\sBB_k$ is 
strictly decreasing for $p \in (0, 2M/(2M+1)]$.
In the proof we will use the following relation between $\bar{\lambda}:=\lambda^*/\sigma^2$
and $p$ a few times:
$\blam=\ln \frac{1-p}{p/(2M)}$ (see \eqref{eq:lambda}).

As in the proofs for earlier results, 
we consider a general $M$ here, but the proof
can be applied to the special case that $M=1$ by some simplification.

From \eqref{eq:j0exp},
we see that $j_k$  either remains the same or is decreasing when $p$ 
is increasing from $0^+$ to $2M/(2M+1)$.
Let $p^{(i)}$ be as in \eqref{eq:piexp}.
When  $p \in (p^{(i)}, p^{(i-1)}]$ or equivalently 
$\blam \in [\lambda^{(i-1)},\lambda^{(i)})$,
$j_k=i$ (see \eqref{eq:jkpw1}).
Since $\rho^\sRB_k$ is continuous with respect to $p$ by Lemma \ref{le:rlpcts},
it is easy to prove that $\rho^\sRB_k-\rho^\sBB_k$ is continuous with respect to $p$.
For any interval $(p^{(i)},p^{(i-1)}]$, 
in the following we would like to show that $\rho^\sRB_k-\rho^\sBB_k$ is decreasing
in the interval.

To simplify notation, let $\eta:=(2j_k-1)\gamma$.
Suppose that $p\in (p^{(i)},p^{(i-1)}]$, 
or equivalently $\blam  \in [\blam^{(i-1)},\blam^{(i)})$.
From \eqref{eq:rbksp}, \eqref{eq:ef} and \eqref{eq:lambda},
\al{
\hspace*{-2mm} 
\frac{\partial \rho^\sRB_k}{\partial p}
 =& \frac{1}{2M} + \frac{M-j_k}{M} \erf(\gamma) 
 -\frac{1}{2M}\erf\left(\frac{{\bar\lambda}}{2\eta}-\frac{1}{2}\eta\right)  
  + \frac{1}{2\sqrt{\pi}M(1-p)\eta}
e^{-\left(\frac{{\bar\lambda}}{2\eta}-\frac{1}{2}\eta\right)^2}
\notag\\
 & - \erf\left(\frac{{\bar\lambda}}{2\eta} + \frac{1}{2}\eta\right) -  \frac{1}{\sqrt{\pi}p \eta}
e^{-\left(\frac{{\bar\lambda}}{2\eta} + \frac{1}{2}\eta\right)^2}. \label{eq:drbsp}
}
Then from \eqref{eq:dbbsp} and \eqref{eq:drbsp} we obtain
\al{
& \frac{\partial (\rho^\sRB_k- \rho^\sBB_k)}{\partial p}  \notag \\
\hspace*{-3mm} =\  & \frac{1-j_p}{M}\erf(\gamma) +\frac{2M-1}{2M}\erf\left(\frac{1}{2}\gamma\right)  
 + \frac{1}{2M}\erf\left(\frac{1}{2}\eta - \frac{{\bar\lambda}}{2\eta} \right)  - \erf\left(\frac{1}{2}\eta + \frac{{\bar\lambda}}{2\eta} \right) \notag \\
& + \frac{1}{2\sqrt{\pi}M(1-p) \eta} e^{-\left(\frac{1}{2}\eta - \frac{\bar\lambda}{2\eta}\right)^2}  
-\frac{1}{\sqrt{\pi}p \eta} e^{-\left(\frac{1}{2}\eta + \frac{{\bar\lambda}}{2\eta} \right)^2}.
\label{eq:drbbb}
}
Using $\blam=\ln \frac{1-p}{p/(2M)}$, we have 
$$
e^{-(\frac{1}{2}\eta - \frac{\blam}{2\eta})^2} = \frac{1-p}{p/(2M)}\,e^{-(\frac{1}{2}\eta + \frac{\blam}{2\eta})^2}.
$$
Therefore the last line in equation \eqref{eq:drbbb} is zero.
Furthermore, since $j_k\geq 1$, $\frac{1-j_k}{M}\erf(\gamma) \leq 0$. 
Thus, to show $\frac{\partial (\rho^\sRB- \rho^\sBB)}{\partial p} < 0$,
from \eqref{eq:drbbb} we observe that it suffices to show
\al{
& \frac{2M-1}{2M} \erf\left(\frac{1}{2}\gamma\right) + \frac{1}{2M}\erf\left(\frac{1}{2}\eta - \frac{\blam}{2\eta} \right) 
 < \erf\left(\frac{1}{2}\eta + \frac{\blam}{2\eta}\right). \label{eq:dspineq}
}
Since $\eta = (2j_k-1)\gamma$ and $j_k\geq 1, \eta \geq \gamma > 0$. 
Thus we have 
\aln{
  \erf\left( \frac{1}{2}\gamma \right) < \erf\left(\frac{1}{2}\eta + \frac{\blam}{2\eta}\right),  \ \
  \erf\left(\frac{1}{2}\eta - \frac{\blam}{2\eta}\right)  < \erf\left(\frac{1}{2}\eta + \frac{\blam}{2\eta}\right).
}
Note that here the two inequalities are strict because $\blam>0$ under the assumption
that $p<2M/(2M+1)$.
Using these two inequalities, we can easily derive \eqref{eq:dspineq}. This concludes our proof.
\end{proof}

\begin{remark}\label{rm:spratiomono}
As stated in the proof, $\rho_k^\sBB$ is increasing 
for $p\in (0,2M/(2M+1))$, so is $\Pr(\x^\sBB=\x^*)$. 
Although the ratio $\Pr(\x^\sRB=\x^*)/\Pr(\x^\sRB=\x^*)$ is decreasing
for $p\in (0,2M/(2M+1))$
and it is equal to 1 when $p=2M/(2M+1)$,
it is easy to find an example for which $\Pr(\x^\sRB=\x^*)$ is not
always decreasing for $p\in (0,2M/(2M+1))$.
\end{remark}

\begin{corollary}
Under the same condition as Theorem \eqref{th:spratiomono},
we have
\be
1 
\leq 
\frac{\Pr(\x^\sRB=\x^\ast)}{\Pr(\x^\sBB=\x^\ast)}
< \frac{1}{\prod_{k=1}^n \erf\left(\frac{r_{kk}}{2\sqrt{2}\sigma}\right)},
\label{eq:rbbbbounds}
\ee
where the bounds are (nearly) attainable with respect to $p$.  
\end{corollary}

\begin{proof}

When $p \rightarrow 0^+$, by \eqref{eq:lambda} $\lambda^*\rightarrow \infty$;
then from \eqref{eq:rbksp}  we observe $\rho_k^\sRB \rightarrow 1$,
and from \eqref{eq:bbksp} we observe
$\rho_k^\sBB \rightarrow \erf\left(\br_{kk}/2\right)$.
Thus from  \eqref{eq:rbsp} and  \eqref{eq:bbsp} we obtain
\al{
\lim_{p \rightarrow 0^+} 
\frac{\Pr(\x^\sRB=\x^\ast)}{\Pr(\x^\sBB=\x^\ast)}
= \frac{1}{\prod_{k=1}^n \erf\left(\frac{r_{kk}}{2\sqrt{2}\sigma}\right)}.
 \label{eq:xrblim1}
}
When $p=2M/(2M+1)$, $\lambda^*=0$ and  $\x^\sRB$ becomes  $\x^\sBB$,
leading to $\Pr(\x^\sRB=\x^\ast)=\Pr(\x^\sBB=\x^\ast)$.
Therefore, by Theorem \ref{th:spratiomono} we have \eqref{eq:rbbbbounds},
where the lower bound is attainable
and the upper bound is nearly attainable with respect to $p$.
\end{proof}

\subsection{A bound on the SP of the $L_0$-regularized Babai point}\label{sec:spbound}
In this section, we give a bound on the SP of $\x^\sRB$, which is an upper bound and a lower bound under different conditions respectively. 
The bound is invariant with respect to the column permutations of $\A$
and it will help us to understand the limit on SP that column permutations 
can achieve, as well as to provide us a motivation for the application of LLL-P strategy.
The approach we use in our proof can be regarded as an 
extension of that used in \cite{WenC17}, which deals with an un-regularized problem.
Here the proof is more complicated due to the discontinuity issue.

To prepare for the proof of our theorem, we first define our target function and introduce a lemma. 
We are interested in the effect of $\br_{kk}$ on $\rho^\sRB_k$ in \eqref{eq:rbksp}. 
For convenience, we use $\gamma$ to replace $\br_{kk}$
in \eqref{eq:rbksp} and rewrite $\rho_k^\sRB$ as $\rho(\gamma)$:
\al{ 
\rho(\gamma) = \ & \frac{p}{2M} +  \frac{M-j_\gamma}{M} p \, \erf\left(\gamma\right) 
    - \frac{p}{2M} \, \erf\left(\frac{1}{2(2j_\gamma-1)}\frac{\bar\lambda}{\gamma}-\frac{1}{2}(2j_\gamma-1)
 \gamma\right)  \notag  \\
  & + (1-p)\, \erf\left(\frac{1}{2(2j_\gamma-1)}\frac{\bar\lambda}{\gamma}+\frac{1}{2}(2j_\gamma-1)\gamma\right),
\label{eq:rho}
}
where we have abused the notation a little bit by replacing  the original $j_k$ in \eqref{eq:rbksp} by $j_\gamma$ to emphasize its dependence on $\gamma$.

For later uses, we define
\begin{equation}\label{eq:fzeta}
    F(\zeta) := \ln\left( \rho (e^\zeta) \right) .
\end{equation}
It is easy to verify that when $\zeta\neq \ln\gamma^{(i)}$ (see \eqref{eq:gj}), 
\al{
& F'(\zeta) = \frac{\rho'(e^\zeta)}{\rho(e^\zeta)}e^\zeta, \label{eq:dF}\\
& F''(\zeta) = \frac{\rho''(e^\zeta)}{\rho(e^\zeta)}e^{2\zeta} + \frac{\rho'(e^\zeta)}{\rho(e^\zeta)} e^\zeta - \left( \frac{\rho'(e^\zeta)}{\rho(e^\zeta)} e^\zeta \right)^2. \label{eq:ddF}
}

\begin{lemma} \label{le:ddF_properties}
Let the regularization parameter $\lambda= \lambda^*$ defined in \eqref{eq:lambda}.
Let $\gamma^{(i)}$ be defined by \eqref{eq:gj} (or equivalent \eqref{eq:giexp}) and \eqref{eq:g0m}.
Suppose $\rho$ and $F$ are given by  \eqref{eq:rho} and \eqref{eq:fzeta},
respectively. 
Then
\ben
    \item $F'(\zeta)$ is continuous on $(-\infty,\infty)$. \label{le:F1der}
    \item $F''(\zeta)\to0^+$ as $\zeta\to -\infty$, and $F''(\zeta)\to0^-$ as $\zeta\to+\infty$.
    \label{le:F2derend}
    \item When $M=1$, $F''(\zeta)$ is continuous on $(-\infty,\infty)$. \label{le:M1F2der}
    \item When $M \geq 2$, for each $i = 1,\ldots,M-1$, 
    $F''(\zeta)$ is continuous on the interval $(\ln \gamma^{(i-1)},\ln \gamma^{(i)})$,
    where $\ln \gamma^{(0)}:=-\infty$ and 
    $\ln \gamma^{(M)}:=\infty$, and
$\lim_{\zeta \to \ln\gamma^{(i)-}} F''(\zeta) < \lim_{\zeta \to \ln\gamma^{(i)+}} F''(\zeta)$.
\label{le:M2F2der}
\item 
Let 
\begin{equation}\label{eq:mu12}
\begin{aligned}
\mu_1 = \min \{\mu: F''(\ln\mu) = 0\}, \\
\mu_2 = \max \{\mu: F''(\ln\mu) = 0\},
\end{aligned}
\end{equation}
i.e. the exponential of the smallest and largest roots of $F''$.  
$F'(\zeta)$ is strictly increasing on $(-\infty, \ln\mu_1)$, and strictly decreasing on $(\ln \mu_2, \infty)$. \label{le:dFmono}
\een
\end{lemma}

\begin{proof}
See Appendix \ref{app:le:ddF_properties}. 
\end{proof}

\begin{theorem}\label{th:mingapsp}
Let the regularization parameter $\lambda = \lambda^\ast$ defined by \eqref{eq:lambda},
and let $\mu_1, \mu_2$ be defined by \eqref{eq:mu12}. 
Denote 
\aln{
& \omega = \Big(\prod_{k=1}^n r_{kk} \Big)^{1/n}, \ \ 
\bomega = \frac{\omega}{\sqrt{2}\sigma}.
}
\begin{enumerate}
    \item \label{th:lb} If $\max_{1\leq k \leq n} r_{kk} \leq \sqrt{2}\sigma\mu_1$, then 
    \be 
    \Pr(\x^\sRB=\x^\ast) \geq \left( \rho (\bomega) \right)^n.
    \label{eq:lb}
    \ee
    \item \label{th:ub}
    If $\min_{1\leq k \leq n} r_{kk} \geq \sqrt{2}\sigma\mu_2$, then \be 
    \Pr(\x^\sRB=\x^\ast) \leq \left( \rho (\bomega) \right)^n.
    \label{eq:ub}
    \ee
\end{enumerate}
In \eqref{eq:lb} and \eqref{eq:ub} the equality holds if and only if $r_{kk}=\omega$ for all $k=1,2,\ldots,n$.
\end{theorem}

\begin{proof} 
By Theorem \ref{th:babaisp} it is easy to see that when $r_{kk} = \omega$ for all $k=1,2,\ldots,n$, 
\eqref{eq:lb} and \eqref{eq:ub} hold with equality. 
Therefore we only need to show that \eqref{eq:lb} and \eqref{eq:ub} hold with 
strict inequality when there exists $j\neq k$ such that $r_{jj}\neq r_{kk}$, 
under the corresponding conditions, respectively. 

Let $\tau_k := \ln(r_{kk}/(\sqrt{2}\sigma))$ for $k=1,2,\ldots,n$ 
and $\tau := \frac{1}{n} \sum_{k=1}^n \tau_k$. 
Then,  by  Theorem \ref{th:babaisp}, to prove our first claim,
it suffices to prove that 
when $\max_{1\leq k \leq n} r_{kk} \leq \sqrt{2}\sigma\mu_1$,
\begin{equation*}
    \frac{1}{n} \sum_{k=1}^n F(\tau_k) > F(\tau),
\end{equation*}
and to prove our second claim, it suffices to prove when $\min_{1\leq k \leq n} r_{kk} \geq \sqrt{2}\sigma\mu_2$, 
\begin{equation*}
    \frac{1}{n} \sum_{k=1}^n F(\tau_k) < F(\tau).
\end{equation*}

Since the condition $\max_{1\leq k \leq n} r_{kk} \leq \sqrt{2}\sigma\mu_1$
is equivalent to $\max_{1\leq k \leq n} \tau_k \leq \ln \mu_1$
and the condition $\min_{1\leq k \leq n} r_{kk} \geq \sqrt{2}\sigma\mu_2$
is equivalent to $\min_{1\leq k \leq n} \tau_k \geq \ln \mu_2$,
we only require that $F(\zeta)$ is strictly convex on the domain $(-\infty, \ln\mu_1)$ 
and strictly concave on the domain $(\ln\mu_2, \infty)$, 
which can be easily obtained from Lemma \ref{le:ddF_properties}-\ref{le:dFmono}. 
This concludes our proof.  
\end{proof}

\begin{remark}
In practice, we can obtain $\ln \mu_1$ and $\ln \mu_2$ 
using a numerical method, e.g., Newton's method,
for solving a nonlinear equation. 
Although $F''$ is not globally differentiable 
as it is not continuous at $\ln \gamma^{(i)}, i=1,\ldots,M-1$, 
it is piecewise differentiable, 
and we can apply a numerical method on each interval to find the roots.
To avoid a possible numerical problem caused by the discontinuities 
we suggest first plotting $F''(\zeta)$ and then choosing two initial points
near to the two roots $\ln \mu_1$ and $\ln \mu_2$, respectively,
for the corresponding iterations. 
Although theoretically the values $\mu_1$ and $\mu_2$ can be anywhere on $(0,\infty)$, in practice they usually do not have a large magnitude. 
Some examples are presented in Table \ref{t:ddfrootsnum}.
We used the MATLAB built-in function \verb!fzero! to find the two roots.
Note that from \eqref{eq:rho} and \eqref{eq:j0exp}, 
we can observe that $F''(\zeta)$ depends on only $\blam^*=\lambda^*/\sigma^2$,
which is determined by $M$ and $p$ (see \eqref{eq:lambda}).
\begin{table}[h!]
\caption{Numerical values of $\mu_1$ and $\mu_2$ under various settings.}\label{t:ddfrootsnum}
\centering
\begin{tabular}{|c|c|c|c|}
\hline
$M$ & $p$ & $\mu_1$   & $\mu_2$   \\ \hline
2   & 0.1 & 0.7182 & 2.0972  \\ \hline
2   & 0.7 & 0.3865 & 0.6913  \\ \hline
4   & 0.1 & 0.3304 & 2.2065  \\ \hline
4   & 0.7 & 0.4637 & 0.6775 \\ \hline
32  & 0.1 & 0.4959 & 0.7072 \\ \hline
32  & 0.7 & 0.4699 & 0.5491 \\ \hline
\end{tabular}
\end{table}
\end{remark}

\begin{remark}
In Theorem \ref{th:mingapsp} the quantity $\omega$ is invariant with respect to column permutations in the QR factorization of $\A$ (see \eqref{eq:pqr}).
In fact, suppose we find the $L_0$-regularized Babai point based on \eqref{eq:l0rbils-p}.
Then the corresponding 
$$
\omega = \Big(\prod_{k=1}^n \hr_{kk} \Big)^{1/n} 
=\det(\hbR)^{1/n} = \det(\A^\top\A)^{1/(2n)},
$$
which is independent of the permutation matrix $\P$.

When $\sigma$ is large, it is likely the condition in 
Theorem \ref{th:mingapsp}-\ref{th:lb} holds, then we have the lower bound.
When $\sigma$ is small, it is likely the condition in Theorem
\ref{th:mingapsp}-\ref{th:ub} holds, then we we have the upper bound.
In applications with small $\sigma$, 
we would like to find good column permutations 
so that the SP of the $L_0$-regularized Babai point is as close to 
the upper bound as possible.
We will investigate column permutations in the next two sections.
\end{remark}

\section{Effects of LLL-P on SP of $\x^\sRB$}
\label{sec:permut}

Theorem \ref{th:mingapsp} in Section \ref{sec:spbound} 
shows that if we can find a permutation matrix $\P$ 
such that all diagonal elements $\hr_{kk}$, $k=1:n$ 
of $\hbR$ in \eqref{eq:pqr} are equal then 
the bound on the SP of $\x^\sRB$ can be reached. 
Such a permutation matrix $\P$ may not exist for a given $\A$,
e.g., the columns of $\A$ are orthogonal to each other, giving
a diagonal $\R$.
But if a column permutation strategy (e.g., LLL-P) 
can decrease the gap 
$\max_{1\leq k \leq n}r_{kk} - \min_{1\leq k \leq n}r_{kk}$,
we can imagine that it will make the SP of $\x^\sRB$ closer to the bound.
In this section, we show how LLL-P increases or decreases the SP of 
regularized Babai point (i.e., the SP is closer to the bound)
under different conditions.
The results to be given can be regarded as extensions of the 
corresponding work in \cite{WenC17}, which deals with an unregularized problem.

In the following we introduce a lemma, which will be used 
for a  theoretical analysis later.
To prepare for it, as in Section \ref{sec:spbound}, 
we replace $\br_{kk}$ by $\gamma$
and rewrite $\rho^\sRB_k$ in \eqref{eq:bbksp}
as $\rho(\gamma)$ (see \eqref{eq:rho}).
For a given $\beta>0$, we define 
\al{
\Phi(\gamma,\beta) := \rho(\gamma)\cdot \rho(\beta/\gamma). \label{eq:Phi}
}

\begin{lemma}\label{le:phimono}
Let the regularization parameter $\lambda = \lambda^\ast$ (see \eqref{eq:lambda}).
Let $\mu_1$ and $\mu_2$ be defined as in \eqref{eq:mu12}. 
Then, for a given $\beta>0$,
when either $\frac{\beta}{\mu_1} < \gamma < \sqrt{\beta}$ or $\sqrt{\beta} < \gamma < \frac{\beta}{\mu_2}$, $\Phi(\gamma,\beta)$ is a strictly decreasing function of $\gamma$.
\end{lemma}

\begin{proof} See Appendix \ref{app:le:phimono}.
\end{proof}

\begin{lemma}\label{le:lllbound1}
Let the regularization parameter $\lambda = \lambda^\ast$ (see \eqref{eq:lambda}). 
Suppose that \eqref{eq:lm} is transformed into \eqref{eq:rlm} using the QR factorization \eqref{eq:qr}, and $\delta\, r_{k-1,k-1}^2 > r_{k-1,k}^2 + r_{kk}^2$. Then we perform 
permutation of columns $k-1$ and $k$ of $\R$ and triangularization as in \eqref{eq:llltri}, obtaining \eqref{eq:permuted-lm}. Suppose also that $\mu_1, \mu_2$ are defined by \eqref{eq:mu12}.  
\begin{enumerate}
\item If $r_{k-1,k-1} \leq \sqrt{2}\sigma \mu_1$, then after the permutation,
\begin{equation}\label{eq:lll-prob-ineq1}
\Pr(\x^\sRB = \x^\ast) \geq \Pr(\z^\sRB = \z^\ast).
\end{equation}
\item If $r_{kk}\geq \sqrt{2}\sigma \mu_2$, then after the permutation, 
\begin{equation}\label{eq:lll-prob-ineq2}
\Pr(\x^\sRB = \x^\ast) \leq \Pr(\z^\sRB = \z^\ast).
\end{equation}
\end{enumerate}
Furthermore, the equalities in \eqref{eq:lll-prob-ineq1} and \eqref{eq:lll-prob-ineq2} hold  if and only if $r_{k-1,k} = 0$.
\end{lemma}

\begin{proof}
When $r_{k-1,k} = 0$, the permutation only exchanges the positions of the diagonal elements, and by Theorem \ref{th:babaisp} the SP stays unchanged, thus the equalities in \eqref{eq:lll-prob-ineq1} and \eqref{eq:lll-prob-ineq2} hold. 

Now we assume $r_{k-1,k}\neq 0$, and we will show that the strict inequalities in \eqref{eq:lll-prob-ineq1} and \eqref{eq:lll-prob-ineq2} hold. 

Define
\begin{equation*}
    \beta := \frac{r_{k-1,k-1}r_{kk}}{2\sigma^2} = \frac{\hr_{k-1,k-1}\hr_{kk}}{2\sigma^2},
\end{equation*}
where the second equality holds by \eqref{eq:lllprop}. Then
\aln{
\min\left\{ \frac{r_{k-1,k-1}}{\sqrt{2}\sigma}, \frac{r_{kk}}{\sqrt{2}\sigma} \right\}
\leq
\sqrt{\beta} 
\leq \max\left\{ \frac{r_{k-1,k-1}}{\sqrt{2}\sigma}, \frac{r_{kk}}{\sqrt{2}\sigma} \right\} .
}
Replacing $r_{k-1,k-1}$ and $r_{kk}$ with $\hr_{k-1,k-1}$ and $\hr_{kk}$ respectively, the above inequalities still hold.
Using the condition that $\delta\, r_{k-1,k-1}^2 > r_{k-1,k}^2 + r_{kk}^2$ 
with $\delta \in (1/4,1]$ and the equalities in \eqref{eq:lllprop}, we can obtain
\al{
    \frac{\beta}{r_{k-1,k-1}/(\sqrt{2}\sigma)} &= \min\left\{ \frac{r_{k-1,k-1}}{\sqrt{2}\sigma}, \frac{r_{kk}}{\sqrt{2}\sigma} \right\}  
     < \min\left\{ \frac{\hr_{k-1,k-1}}{\sqrt{2}\sigma}, \frac{\hr_{kk}}{\sqrt{2}\sigma} \right\} \leq \sqrt{\beta}, \label{eq:betaineqmin} \\
    \sqrt{\beta} &\leq \max\left\{\frac{\hr_{k-1,k-1}}{\sqrt{2}\sigma}, \frac{\hr_{kk}}{\sqrt{2}\sigma}\right\}  
     <\max\left\{\frac{r_{k-1,k-1}}{\sqrt{2}\sigma}, \frac{r_{kk}}{\sqrt{2}\sigma}\right\} = \frac{\beta}{r_{kk} / (\sqrt{2}\sigma)}. \label{eq:betaineqmax}
}

First, we assume $r_{k-1,k-1} \leq \sqrt{2}\sigma \mu_1$. Since the permutation only affects columns $k$ and $k-1$, other diagonal elements remain the same. By Theorem \ref{th:babaisp}, to prove \eqref{eq:lll-prob-ineq1} it suffices to show
\begin{equation*}
\begin{aligned}
     \rho(r_{k-1,k-1}/(\sqrt{2}\sigma))\cdot \rho(r_{k,k}/(\sqrt{2}\sigma))  \geq \rho(\hr_{k-1,k-1}/(\sqrt{2}\sigma))\cdot \rho(\hr_{k,k}/(\sqrt{2}\sigma)),
\end{aligned}
\end{equation*}
which is equivalent to
\begin{equation}\label{eq:phiineq1}
\begin{aligned}
     \Phi\left(\min\left\{ \frac{r_{k-1,k-1}}{\sqrt{2}\sigma}, \frac{r_{kk}}{\sqrt{2}\sigma} \right\}, \beta\right)  
      \geq \Phi\left(\min\left\{ \frac{\hr_{k-1,k-1}}{\sqrt{2}\sigma}, \frac{\hr_{kk}}{\sqrt{2}\sigma} \right\}, \beta\right)    .
\end{aligned}
\end{equation}
When $r_{k-1,k-1} \leq \sqrt{2}\sigma \mu_1$, then $\frac{\beta}{\mu_1} \leq \frac{\beta}{r_{k-1,k-1}/(\sqrt{2}\sigma)}$. Thus from \eqref{eq:betaineqmin}, 
\begin{equation*}
\begin{aligned}
    \frac{\beta}{\mu_1} &\leq \min\left\{ \frac{r_{k-1,k-1}}{\sqrt{2}\sigma}, \frac{r_{kk}}{\sqrt{2}\sigma} \right\}  
    < \min\left\{ \frac{\hr_{k-1,k-1}}{\sqrt{2}\sigma}, \frac{\hr_{kk}}{\sqrt{2}\sigma} \right\} \leq \sqrt{\beta}    ,
\end{aligned}
\end{equation*}
and by Lemma \ref{le:phimono} the strict inequality in \eqref{eq:phiineq1} holds. 

Now we assume $r_{kk}\geq \sqrt{2}\sigma \mu_2$. 
Using the same reasoning as before, to prove \eqref{eq:lll-prob-ineq2} it suffices to show
\begin{equation}\label{eq:phiineq2}
\begin{aligned}
      \Phi\left(\max\left\{ \frac{r_{k-1,k-1}}{\sqrt{2}\sigma}, \frac{r_{kk}}{\sqrt{2}\sigma} \right\}, \beta\right)  
      \leq \Phi\left(\max\left\{ \frac{\hr_{k-1,k-1}}{\sqrt{2}\sigma}, \frac{\hr_{kk}}{\sqrt{2}\sigma} \right\}, \beta\right).   
\end{aligned}
\end{equation}
When $r_{kk}\geq \sqrt{2}\sigma \mu_2$, we have $\frac{\beta}{r_{kk}/(\sqrt{2}\sigma)} \leq \frac{\beta}{\mu_2}$. Thus from \eqref{eq:betaineqmax}, 
\begin{equation*}
\begin{aligned}
\sqrt{\beta} &\leq \max\left\{\frac{\hr_{k-1,k-1}}{\sqrt{2}\sigma}, \frac{\hr_{kk}}{\sqrt{2}\sigma}\right\}  
 < \max\left\{\frac{r_{k-1,k-1}}{\sqrt{2}\sigma}, \frac{r_{kk}}{\sqrt{2}\sigma}\right\} \leq \frac{\beta}{\mu_2}.        
\end{aligned}
\end{equation*}
By Lemma \ref{le:phimono} the strict inequality in \eqref{eq:phiineq2} holds. 
\end{proof}

Based on Lemma \ref{le:lllbound1}, we can obtain the following general result for applying LLL-P strategy on $L_0$-RBILS problems.

\begin{theorem}\label{th:lllbound2}
Suppose that \eqref{eq:lm} is transformed into \eqref{eq:rlm} using the QR factorization \eqref{eq:qr}, then into \eqref{eq:permuted-lm} using the QR factorization \eqref{eq:pqr} where LLL-P strategy is used for column permutations. Suppose also that $\mu_1, \mu_2$ are defined as in \eqref{eq:mu12}. Then:
\begin{enumerate}
    \item If $\R$ satisfies $\max_{1\leq k\leq n} r_{kk}\leq \sqrt{2}\sigma\mu_1$,
    then 
    \begin{equation}\label{eq:lll-gen-ineq1}
    \Pr(\x^\sRB = \x^\ast) \geq \Pr(\z^\sRB = \z^\ast).
    \end{equation}
    \item If $\R$ satisfies $\min_{1\leq k \leq n} r_{kk} \geq \sqrt{2}\sigma \mu_2$, 
    then 
    \begin{equation}\label{eq:lll-gen-ineq2}
    \Pr(\x^\sRB = \x^\ast) \leq \Pr(\z^\sRB = \z^\ast).
    \end{equation}
\end{enumerate}
Furthermore, the equalities in \eqref{eq:lll-gen-ineq1} and \eqref{eq:lll-gen-ineq2} hold if and only if no column permutation occurs in the LLL-P process, or whenever two columns $k-1$ and $k$ are permuted, $r_{k-1,k} = 0$.
\end{theorem}

\begin{proof}
By \eqref{eq:betaineqmax} and \eqref{eq:betaineqmin}, 
each time LLL-P performs a column permutation, 
we have 
$\min_{1\leq k \leq n} \hr_{kk} \geq \min_{1\leq k\leq n}r_{kk}$
and $\max_{1\leq k \leq n} \hr_{kk} \leq \max_{1\leq k\leq n}r_{kk}$.
Therefore, the diagonal elements of $\hat{\R}$ after a column permutation satisfy $\min_{1\leq k\leq n} r_{kk}\leq \hat{r}_{kk} \leq \max_{1\leq k \leq n} r_{kk}$ for all $k=1:n$. 
Thus, after each permutation, the conditions are still satisfied. 
We can then apply Lemma \ref{le:lllbound1} to get our conclusion.
\end{proof}

\begin{remark}\label{rm:lll-rkkgap}
When the noise standard deviation $\sigma$ is small, Theorem \ref{th:lllbound2} shows that the inequality \eqref{eq:lll-gen-ineq2} is likely to hold, therefore applying the LLL-P strategy is likely to increase the success probability of the regularized Babai estimator. However when $\sigma$ is large, applying LLL-P will decrease the success probability. Therefore before applying the LLL-P strategy to our linear model, we need to verify if the condition that guarantees an improvement on success probability holds. 
If the diagonal elements do not satisfy either condition in Theorem \ref{th:lllbound2}, i.e. when
\begin{equation}\label{eq:lllgap}
\frac{\max_{1\leq k\leq n}r_{kk}}{\mu_2} 
\leq \sqrt{2}\sigma 
\leq \frac{\min_{1\leq k \leq n}r_{kk}}{\mu_1}  ,  
\end{equation}
the success probability may increase or decrease. 
To demonstrate this, we perform numerical experiments 
under the following setting. 
Let $p=0.4$ and $M=2$. Under this setting, $\mu_1 = 0.5475$ and $\mu_2 = 1.2546$.
We use MATLAB's function $0.5*\text{randn}(4,4)$ to generate 
a model matrix $\A$,
and then we check if the most left hand side of \eqref{eq:lllgap}
is larger than the most right hand side.
If it is not true, we discard the matrix and continue sampling.  
When it is true, we take 
$\sigma = \frac{1}{2\sqrt{2}}\left(\max_{1\leq k\leq n}\frac{r_{kk}}{\mu_2} + \min_{1\leq k \leq n}\frac{r_{kk}}{\mu_1}\right)$
so that \eqref{eq:lllgap} holds.
Table \ref{t:lllspgap} presents the SPs of the $L_0$-regularized Babai point
for six found matrices,
where $P$ and $P_L$ are SPs without and with applying LLL-P, respectively.

\begin{table}
\caption{SP of the $L_0$-regularized Babai point}\label{t:lllspgap}
\centering
\begin{tabular}{|c|c|c|}
\hline
$\sigma$ & $P$ & $P_L$ \\ \hline
0.6054   & 0.3515   & 0.3526     \\ \hline
0.8543   & 0.3377   & 0.3376     \\ \hline
0.4942   & 0.3358   & 0.3360     \\ \hline
0.6137   & 0.3231   & 0.3225     \\ \hline
0.6198   & 0.3601   & 0.3594     \\ \hline
0.5860   & 0.3544   & 0.3544     \\ \hline
\end{tabular}
\end{table}
\end{remark}

\begin{remark} \label{rm:permute}
Theorem \ref{th:lllbound2} focuses on the effect of LLL-P on the SP of $\x^\sRB$ only. Different permutation strategies such as SQRD and V-BLAST may not have the same effect on the SP under the same conditions. Consider the case where $M=2, p=0.5$.
Under this setting, $\mu_1 = 0.4982$ and $\mu_2 = 1.0429$ where $\mu_1, \mu_2$ are defined as in \eqref{eq:mu12}.
Let $P,P_L,P_S,P_V$ denote the SP of the $L_0$-regularized Babai points corresponding to the original $\R$ and the ones transformed by the strategies LLL-P, SQRD and V-BLAST, respectively.

Consider the case where
\begin{equation*}
    \R = \begin{bmatrix}
3.5 & -4 & -3 \\
 & 0.5 & -2  \\
 &  & 0.5
    \end{bmatrix}, \quad \sigma = 0.2.
\end{equation*}
In this case, $\min_{1\leq k \leq n}r_{kk} \geq \sqrt{2}\sigma\mu_2 = 0.2950$. The success probabilities are given by
\begin{equation*}
    P = 0.8109,\,  P_L = P_V = 0.8491,\, P_S = 0.6033.
\end{equation*}
The LLL-P strategy increases the SP, which coincides with Theorem \ref{th:lllbound2}. The V-BLAST strategy increases the SP as well, while the SQRD strategy decreases the SP.

Now, consider the case where
\begin{equation*}
    \R = \begin{bmatrix}
    1 & -1.5 & 2 \\
    & 0.8 & -1 \\
    &  & 0.4 
    \end{bmatrix}, \quad \sigma = 4.0.
\end{equation*}
In this case,
$\max_{1\leq k \leq n} r_{kk} \leq \sqrt{2}\sigma\mu_1 = 2.8180$. The success probabilities are given by
\begin{equation*}
    P = P_L = P_S = 0.1264, \, P_V = 0.1348.
\end{equation*}
The LLL-P strategy does not increase the SP (the permuted ordering is identical to the original one), which coincides with Theorem \ref{th:lllbound2}, and the SQRD strategy has the same effect, while the V-BLAST strategy increases the SP. 
\end{remark}

\section{Success probability based column permutation strategies}\label{sec:spgalg}

In this section, we propose three SP based column permutation strategies to increase the SP of the Babai point $\x^\sRB$.
Specially, we propose a local greedy strategy, a global greedy strategy and a mixed one.

\subsection{A local success probability  based strategy LSP}

In this subsection, we propose a local success probability based strategy,
abbreviated as LSP, 
which turns out to be the V-BLAST permutation strategy. 

We would like to develop a greedy permutation strategy using $\rho_k^\sRB$ (see \eqref{eq:rbksp}) as our measure. 
To determine the last column,
for $j=1,2,\ldots,n$, we interchange columns $j$ and $n$ of $\R$ (when $j=n$, $\R$ is just
the original one)
and then return the matrix to upper triangular 
by a series of Givens rotations applied from the left. 
Denote the resulting upper triangular matrix by $\R^{(n,j)}$.
Then we can obtain  $\rho_n^\sRB$
by the formula \eqref{eq:rbksp} (i.e., the SP of the $L_0$-regularized Babai point at level $n$). 
We rewrite it as $\rho_n^\sRB(r_{nn}^{(n,j)})$ as it is a function of $r_{nn}^{(n,j)}$.
We would like to find column $j$ such that $\rho_n^\sRB(r_{nn}^{(n,j)})$ is maximal, i.e.,
\be  \label{eq:col}
p_n = \argmax_{j} \rho_n^\sRB(r_{nn}^{(n,j)}).
\ee
Then the $p_n$-th column of $\R$ is put as the last column
and $\R$ is updated to be upper triangular. 
After that we recurse on the submatrix $(r_{ij})_{i,j\in\{1,\ldots,n-1\}}$.
We repeat the process until all
columns of the matrix are reordered.

By Theorem \ref{th:rmono},  $\rho_n^\sRB(r_{nn}^{(n,j)})$ is strictly increasing
for $r_{nn} \in (0,\infty)$. 
Therefore, we have
\be \label{eq:lspn}
p_n := \argmax_{j} r_{nn}^{(n,j)}.
\ee
This indicates that this LSP permutation strategy is just the V-BLAST strategy, see Algorithm \ref{alg:vblast}.

\begin{remark}
The LSP permutation strategy was proposed in the conference paper \cite{ChaP18},
but it was not realized there that it is just the V-BLAST strategy,
because in \cite{ChaP18} the SP formula of $\x^\sRB$
is more complicated and the monotonicity of $\rho_k^\sRB$ with respect to $r_{kk}$
was not established.
\end{remark}

\subsection{A global success probability based strategy GSP}

When determining the $k$-th column of $\R$, the LSP strategy or the V-BLAST strategy chooses a column to 
maximize  the SP of $x_k^\sRB$ at level $k$,
and it ignores 
its impact on the SP of $x_i^\sRB$ for $i<k$.
It turns out that applying LSP  may decrease the SP of $\x^\sRB$.

To guarantee that the SP will increase (not strictly) after the columns of $\R$ are permuted, 
we propose a new strategy based on the SP formula \eqref{eq:rbsp}.

Here we show how to determine the last column. Other columns can be determined similarly.
For $j=1,2,\ldots,n$ we interchange the columns $j$ and $n$ of $\R$, 
and then return the matrix to upper triangular  by a series of Givens rotations applied from the left. 
Denote the resulting upper triangular matrix by $\R^{(n,j)}$.
Then we can obtain   $P^\sRB(\R^{(n,j)})$, 
the SP of the $L_0$-regularized Babai point corresponding to  $\R^{(n,j)}$,  see \eqref{eq:rbsp}.
We would like to find column $p_n$ such that  $P^\sRB(\R^{(n,p_n)})$ is maximal, i.e.,
\be \label{eq:gspn}
p_n = \argmax_{j} P^\sRB(\R^{(n,j)}).
\ee
Note that 
$$
P^\sRB(\R^{(n,p_n)}) \geq P^\sRB(\R^{(n,n)}),
$$
where $\R^{(n,n)}$ is the original $\R$. 
Thus by interchanging columns $p_n$ and $n$, the SP of the Babai point increases.
Then we work with the $(n-1)\times (n-1)$ leading principle submatrix of $ \R^{(n,p_n)}$
to determine its last column  in the same way. 
We repeat the process until all columns are reordered.
We refer to this as the global success probability based
permutation strategy (in short GSP),
and the algorithm is described in Algorithm \ref{alg:gsp}.
 \begin{algorithm}
\caption{Algorithm GSP}
\label{alg:gsp}
\begin{algorithmic}[1]
\STATE Set $\R_0=\R$.
\STATE Initialize $\p = [1,2,\ldots,n]$.
\FOR{$k=n,n-1,\ldots,2$}
 	\STATE Compute $\rho=P^\sRB(\R)$. 
	\STATE Set $\R'=\R, j_k = k$.
         \FOR{$j=1,\ldots,k-1$}
                  \STATE Set $\R_{tem}=\R$. 
                  \STATE \label{l:qrupdate} Interchange columns $j$ and $k$ of $\R_{tem}$ and triangularize $\R_{tem}$ by Givens rotations.
                  \STATE Compute $P^\sRB(\R_{tem})$.
                  \IF{$P^\sRB(\R_{tem}) > \rho$}
                  \STATE Update: $\rho= P^\sRB(\R_{tem})$, $j_k=j$, $\R'=\R_{tem}$.
                  \ENDIF
        \ENDFOR
	\STATE Interchange $p_{j_k}$ and $p_k$. 
	\STATE Set $\R= (r'_{st})_{s,t\in \{1,\ldots,k-1\}}$. \label{line:dimred}
\ENDFOR
\STATE Permute columns of $\R_0$ based on $\p$.
\STATE Compute its QR factorization: $\bar{\Q}^\top  \R_0=\R$. \label{line:R0qr}
\end{algorithmic}
\end{algorithm}

The cost of Line \ref{l:qrupdate} dominates the cost of Algorithm \ref{alg:gsp}.
For each $k$ and $j$, its cost is $O(n^2)$ flops.
The whole algorithm costs $O(n^4)$ flops.

Compared with Algorithm LSP (i.e., Algorithm \ref{alg:vblast}),  
Algorithm GSP costs more.
Furthermore, it needs the value of the noise standard deviation $\sigma$, 
while LSP does not.
However, GSP guarantees to increase the SP of the Babai point $\x^\sRB$,
while LSP does not.

\begin{remark}
One can implement Algorithm \ref{alg:gsp} without reducing the dimension of $\R$
at each step (see Line \ref{line:dimred}), and instead perform Givens rotations to the whole $n\times n$ matrix $\R$. 
By doing this, we can avoid performing QR factorization at the end to recover the permuted $\R$ (see Line \ref{line:R0qr}), as Algorithm \ref{alg:vblast} does.
However, this can drastically increase the computational cost because
to choose one column at step $k$, we would apply a sequence of Givens rotations to $\R$ for each $j=1,2,\ldots,k-1$. 
\end{remark}

\subsection{A mixed success probability based strategy MSP}
As we will see later in the numerical experiments in Section \ref{sec:numexp_avgsp_strategies}, the LSP strategy generally outperforms the GSP strategy when $\sigma$ is small. 
However, the former does not guarantee that the SP will increase, while the latter does.
This motivates us to propose a mixed strategy, which has the advantages 
of both LSP and GSP to great extent.

As before, we describe how to determine the $n$-th column here. 
The procedure can be applied recursively to determine other columns. 
First, we use the LSP strategy (as described in Algorithm \ref{alg:vblast}) to find the index $p_n$ (see \eqref{eq:lspn}) 
and interchange columns $p_n$ and $n$ of $\R$. 
We then use Givens rotations to bring $\R$ back to the upper triangular structure.
Denote this transformed upper triangular matrix by $\R^{(n,p_n)}$. 
We then compute the SP $P^\sRB(\R^{(n,p_n)})$ using the formula \eqref{eq:rbsp}. 
Comparing this with the original SP $P^\sRB(\R)$, 
if $P^\sRB(\R^{(n,p_n)})\geq P^\sRB(\R)$, 
column $p_n$ of $\R$ is our choice for the last column and
we move on to the next level to choose column $n-1$. 
Otherwise, we use the GSP strategy to find the index $p_n$ 
defined by \eqref{eq:gspn} and interchange columns $p_n$ 
and $n$ of the origin $\R$,
and then we move on to the next level to choose column $n-1$.
We refer to this as the mixed success probability based strategy 
(in short MSP). 
The algorithm is described in Algorithm \ref{alg:msp}. 

The cost of this algorithm can be of the same order as that of LSP,
i.e., $O(n^3)$ flops.
In the worse case, at every $k=n,n-1,\ldots,2$ the LSP strategy of selecting a column would decrease the SP, then the cost would be the sum of LSP and GSP, 
which is $O(n^4)$ flops.

\begin{algorithm}
\caption{Algorithm MSP}\label{alg:msp}
\begin{algorithmic}[1]
\STATE Initialize $\p = [1,2,\ldots,n], \F = \R^{-\top}$. 
\STATE Create buffer: set $\R_0 = \R_1 = \R, \F_1 = \F, \p_1 = \p$.
\FOR{$k=n,n-1,\ldots,2$}
    \STATE Compute $\rho = P^\sRB(\R)$.
    \STATE Compute $\d$: $d_j = \|(\f_j)_{jk}\|_2$ for $j=1,\ldots,k$.
    \STATE Set $j_k^L = \argmin_{1\leq j \leq k} d_j$. 
    \STATE Interchange $\p(k)$ and  $\p(j_k^L)$, and interchange columns $k$ and $j_k^L$ 
    of both $\R$ and $\F$. 
    \STATE Triangularize $\R$ by Givens rotations: $\R = \G^\top\R$.
    \STATE Update $\F$: $\F = \G^\top \F$.
    \STATE Compute $\rho' = P^\sRB(\R)$. 
    \IF{$\rho' < \rho$}
        \STATE Recover: set $\R = \R_1, \F = \F_1, \p = \p_1$. 
        \STATE Initialize $\R' = \R$, $\G'=\I_k$,  $j_k = k$.
        \FOR{$j = 1,\ldots,k-1$}
            \IF{$j \neq i_k$} 
            \STATE Set $\R_{tem} = \R$.
            \STATE Interchange columns $j$ and $k$ of $\R_{tem}$,   
            triangularize it by Givens rotations: $\R_{tem} = \G_{tem}^\top\R_{tem}$
            \STATE Compute $\rho' = P^\sRB(\R_{tem})$.
            \IF{$\rho' > \rho$}
                \STATE Update: $\rho = \rho'$, $\R' = \R_{tem}$, $\G'=\G_{tem}$,
                $j_k = j$.
            \ENDIF
             \ENDIF
        \ENDFOR
        \STATE Interchange $p_k$ and $p_{j_k}$.
        \STATE Update: set $\R = \R'$, interchange columns $k$ and $j_k$ of $\F$ and set  $\F = \G'^\top\F$.
    \ENDIF
    \STATE 
    Set $\R = (r_{st})_{s,t\in\{1,\ldots,k-1\}}$, 
    $\F = (f_{st})_{s,t\in\{1,\ldots,k-1\}}$.
    \STATE Update buffer: $\R_1 = \R, \F_1 = \F,  \p_1 = \p$.
\ENDFOR
\STATE Permute columns of $\R_0$ based on $\p$.
\STATE Compute its QR factorization: $\bar{\Q}^\top \R_0 = \R$.
\end{algorithmic}
\end{algorithm}

\section{Numerical Tests}\label{sec:tests}
In this section we present numerical experiment results
to demonstrate the theoretical findings in the previous sections
and compare different permutation strategies.

The following settings are applied to all numerical experiments. 
In \eqref{eq:lm} we set $m=n=20$, $M=4$, the regularization parameter $\lambda$ as in \eqref{eq:lambda}, and the model matrices $\A$ are generated with two methods:
\begin{itemize}
    \item \textbf{Type 1}. The entries of $\A$ are sampled independently and randomly from a standard Gaussian distribution. 
    \item \textbf{Type 2}. $\A = \U\D\V^\top$. $\U,\V$ are orthogonal matrices obtained from
    QR factorization of two independent $n\times n$ matrices whose entries are sampled independently and randomly from a standard Gaussian distribution, 
    and $\D$ is an $n\times n$ diagonal matrix with $d_{ii} = 10^{3(n/2 - i) / (n-1)}$. Under this setting, the condition number of $\A$ is always 1000.
\end{itemize}

\subsection{Theoretical and experimental SP of $\x^\sRB$}\label{sec:numexp_the_exp_sp}
First we compare the theoretical SP calculated according to Theorem \ref{th:babaisp} and the experimental SP. 
In the model \eqref{eq:lm} we set $\sigma \in \{0.05,0.10,0.15,\ldots,0.50\}$.
For each $\sigma$, 
we randomly generate 100 model matrices $\A$ by the MATLAB built in function
$\text{randn(n,n)}$, 
and for each $\A$ we generate 100 parameter vectors $\x^\ast$ 
whose entries are independently sampled from distribution \eqref{eq:xprob}. 
For each $\x^\ast$ we again generate 100 noise vectors $\v\sim N(\0,\sigma^2 \I)$
by $\sigma*\text{randn(n,1)}$, 
In total each $\sigma$ corresponds to $10^6$ instances. 
The average experimental SP is calculated as the ratio of number of correct estimates to $10^6$.
The average theoretical SP is computed according to equation \eqref{eq:rbsp} over 100 model matrices.

Figure \ref{fig:sp_model1}-\ref{fig:sp_model2} display 
the average theoretical and experimental SP of $\x^\sRB$ and $\x^\sBB$ 
for the two types of matrices.
In the plots, the abbreviated prefixes ``EXP-" and ``TH-" stand for ``experimental" and ``theoretical" respectively, and ``RBSP" and ``BBSP" stand for
$L_0$-regularized Babai point success probability and (unregularized) box-constrained 
Babai point success probability, respectively. 
From these figures, we observe:
\ben 
\item The theoretical curves almost coincide with their corresponding experimental curves in every setting for both $\x^\sRB$ and $\x^\sBB$.

\item $\x^\sRB$ has higher SP compared to $\x^\sBB$ for the same $p$ and the same
$\sigma$, which agrees with Theorem \ref{th:relation}.

\item For a fixed $\sigma$, the ratio of the SP of $\x^\sRB$ 
(the top curve) to the SP of $\x^\sBB$ (the bottom curve) when $p=0.3$
is larger than the ratio of the SP of $\x^\sRB$ 
(the second curve) to the SP of $\x^\sBB$ (the third curve) when $p=0.7$, 
which agrees with Theorem \ref{th:spratiomono}.

\item All the SPs for Type 1 matrices are smaller than the corresponding SPs
for Type 2 matrices. 
This is because on average the largest $r_{kk}$ for Type 2 matrices are 
more or less the same size as those for Type 1 matrices, 
but the smallest $r_{kk}$ for Type 2 matrices are 
much smaller than those for Type 1 matrices (note that on average Type 2 matrices
are much more ill-conditioned than Type 1 matrices).
\een


\begin{figure}[htbp]
    \subfigure[Type 1 matrices]{%
    \includegraphics[width=0.45\columnwidth]{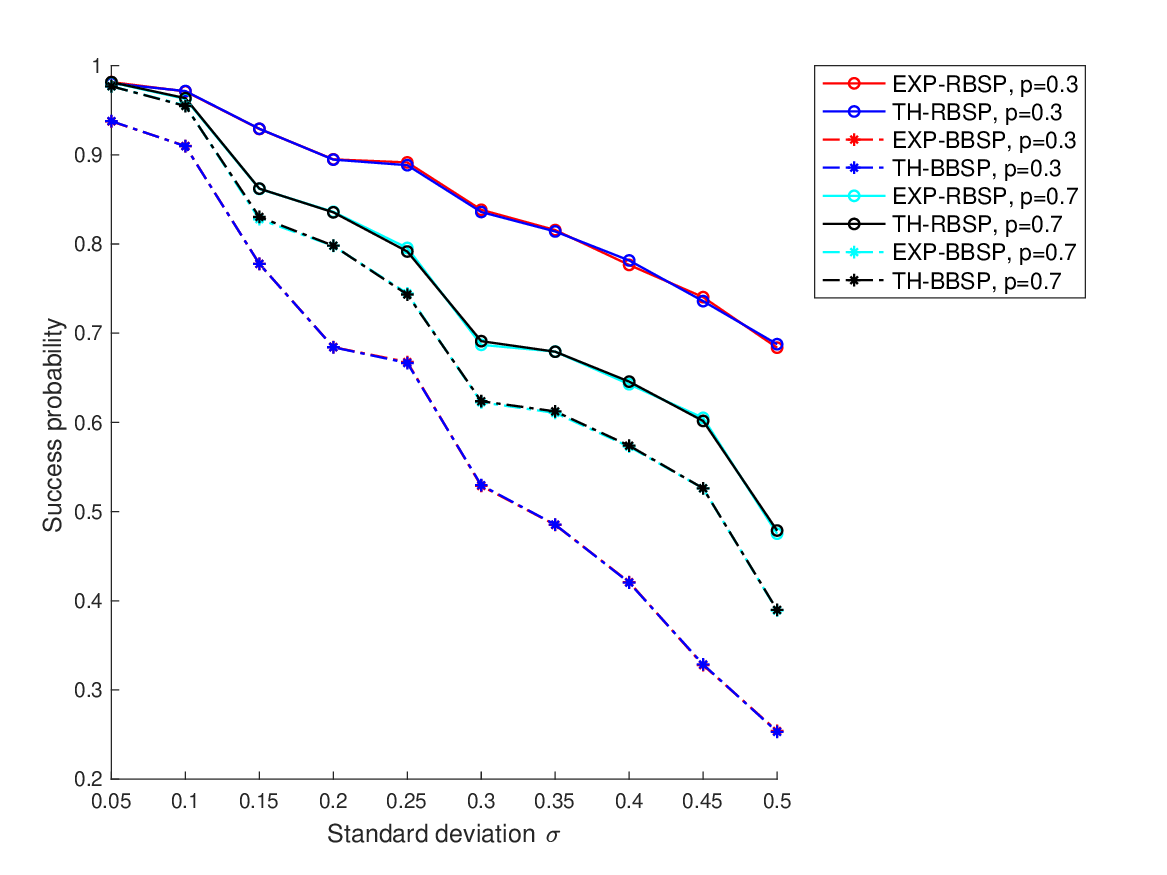}
    \label{fig:sp_model1}
  }
  \subfigure[Type 2 matrices]{%
    \includegraphics[width=0.45\columnwidth]{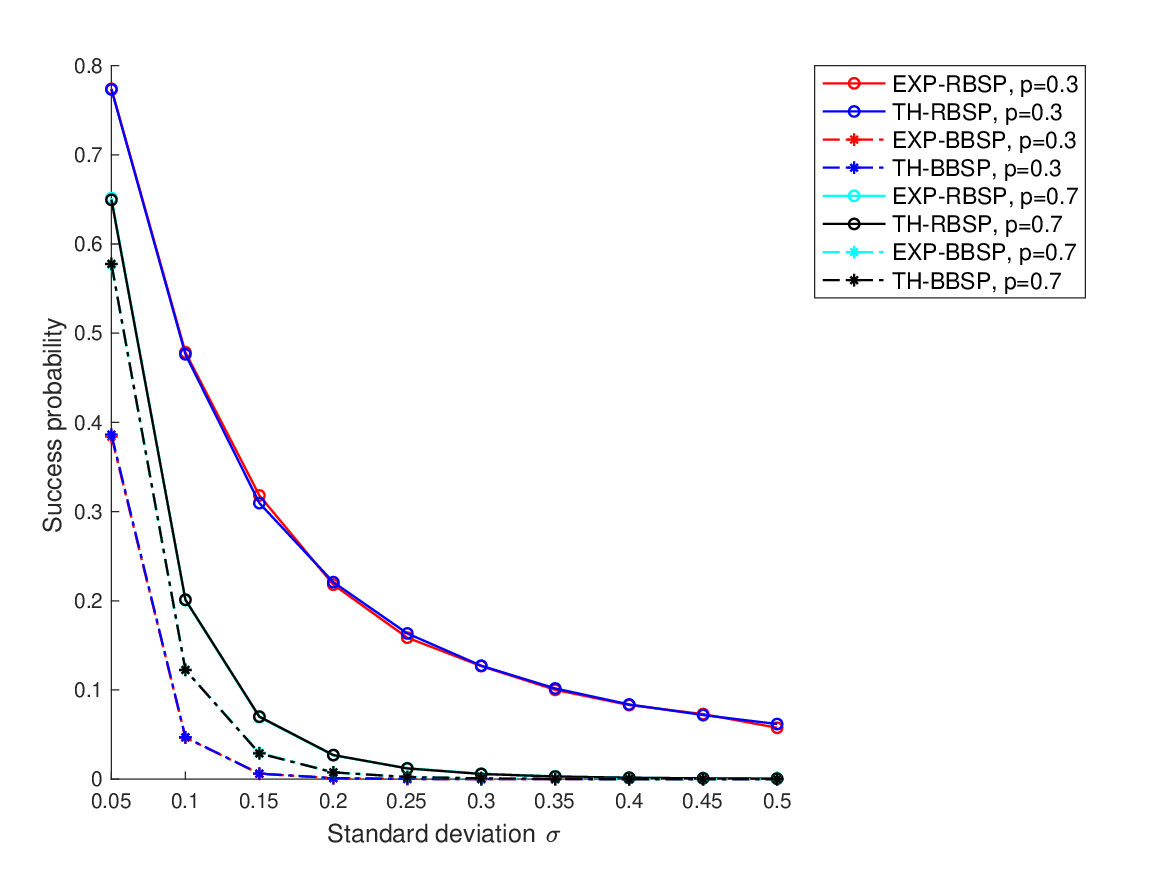}
   \label{fig:sp_model2}
  }
   \centering
  \caption{Theoretical and experimental success probability of the Babai points}
   \label{fig:spplots}
\end{figure}

\subsection{Effect of LLL-P on the SP of $\x^\sRB$ from Theorem \ref{th:lllbound2}}\label{sec:numexp_lllp_spchange}
We would like to verify the result from Theorem \ref{th:lllbound2}. 
We set $p\in\{0.3,0.7\}$. 
For each $p$, we generate 1000 matrices $\A$ for each type.
For each $\A$, we obtain its QR factorization $\A = \Q\R$, 
and we set $\sigma_1 = \min r_{kk} / (2\sqrt{2}\mu_2), \sigma_2 = \max_{1\leq k \leq n} \sqrt{2}r_{kk}/\mu_1$, 
where $\mu_1,\mu_2$ are defined in \eqref{eq:mu12}. 
When $p=0.3$, $\mu_1 = 0.2840$ and $\mu_2=0.6518$;
and when $p=0.7$, $\mu_1=0.4637$ and $\mu_2=0.6775$.
We then apply LLL-P strategy onto the model matrix 
and compare the theoretical SP using the standard deviation
$\sigma = \sigma_1$ and $\sigma = \sigma_2$, respectively, 
before and after the permutation.
By Theorem \ref{th:lllbound2}, we expect the SP to increase when $\sigma = \sigma_1$ and decrease when $\sigma = \sigma_2$.

We also remark that for all the generated model matrices we found 
the most left hand side of \eqref{eq:lllgap} is larger than 
the most right hand side, which means for any $\sigma > 0$
either the condition in conclusion 1 or the condition in conclusion 2
in Theorem \ref{th:lllbound2} will be satisfied. 
Therefore applying the LLL-P strategy will either increase or decrease the SP (not strictly). 
Table \ref{t:lll_sp_change} presents the number of times the SP increases, decreases or stays unchanged after LLL-P is applied. 
When $\sigma = \sigma_1$, LLL-P does not decrease the SP for any $\A$ and $p$, and when $\sigma = \sigma_2$, LLL-P does not increase the SP. This agrees with the conclusions from Theorem \ref{th:lllbound2}.

\begin{table}
\centering
\caption{Changes in the SP of $\x^\sRB$ after applying LLL-P. }
\label{t:lll_sp_change}
\begin{tabular}{|c|c|c|c|c|c|} 
\hline
\multirow{2}{*}{$\sigma$}   & \multirow{2}{*}{SP change} & \multicolumn{2}{c|}{$p=0.3$}    & \multicolumn{2}{c|}{$p=0.7$}     \\ 
\cline{3-6}
                            &                            & Type 1 & Type 2 & Type 1 & Type 2  \\ 
\hline
\multirow{3}{*}{$\sigma_1$} & Strict increase            & 995            & 1000           & 995            & 1000            \\ 
\cline{2-6}
                            & Strict decrease            & 0              & 0              & 0              & 0               \\ 
\cline{2-6}
                            & No change                  & 5              & 0              & 5              & 0               \\ 
\hline
\multirow{3}{*}{$\sigma_2$} & Strict increase            & 0              & 0              & 0              & 0               \\ 
\cline{2-6}
                            & Strict decrease            & 1000           & 1000            & 1000           & 1000            \\ 
\cline{2-6}
                            & No change                  & 0              & 0              & 0              & 0               \\
\hline
\end{tabular}
\end{table}

\subsection{Effects of column permutation strategies on the experimental SP of $\x^\sRB$}\label{sec:numexp_avgsp_strategies}
To see how the column permutation strategies affect 
the SP of the $L_0$-regularized Babai point, 
we test for the experimental SP of $\x^\sRB$ 
with various permutation strategies, 
namely LLL-P, SQRD, LSP (which produces the same ordering as V-BLAST), GSP as in Algorithm \ref{alg:gsp} and MSP as in Algorithm \ref{alg:msp}.

In the model \eqref{eq:lm}, we set the value $p=0.3$,
$\sigma \in \{0.05,0.10,0.50,0.80,1.00,1.50,2.00\}$ for Type 1 matrices, and $\sigma \in \{0.01,0.03,0.05,0.10,0.20,0.30,0.50\}$ for Type 2 matrices. 
For each $\sigma$, we randomly generate 500 model matrices for each type, 
and for each $\A$ we generate 500 parameter vectors $\x^\ast$ 
whose entries are independently sampled from distribution \eqref{eq:xprob}, 
and 500 noise vectors $\v$ sampled from $N(\0, \sigma^2\I)$ by $\sigma*\text{randn(n,1)}$. 
In total each $\sigma$ corresponds to $2.5\times10^5$ instances.
The average experimental SP is calculated as the ratio of correct estimates to $2.5\times 10^5$. 

Table \ref{t:sp_perm_type1}-\ref{t:sp_perm_type2} display the average experimental SP 
of the $L_0$-regularized Babai points corresponding to various column permutation strategies,
where ``No'' means no permutation strategy is employed.

From the tables, 
we can see that applying $L_0$ regularization gives improvement on the SP with or without permutation strategies. The improvement is especially significant when $\sigma$ is large. 
Furthermore, when $\sigma$ is small, all strategies tend to increase the SP of $\x^\sRB$, and LSP has best performance,
closely followed by MSP.
As $\sigma$ becomes large, 
the performance of MSP is not as good as GSP in this case. 
Overall, the performance of MSP tends to be in between LSP and GSP, which means that MSP has better performance than GSP when $\sigma$ is small, and guarantees an SP increase, 
when $\sigma$ is large.

\begin{table*}[]
\centering
\caption{SP of Babai points with column permutations, Type 1.}
\label{t:sp_perm_type1}
\begin{tabular}{|cc|ccccccc|}
\hline
\multicolumn{2}{|c|}{\multirow{2}{*}{$\P$}}              & \multicolumn{7}{c|}{$\sigma$}                                \\ \cline{3-9} 
\multicolumn{2}{|c|}{}                                   & 0.05   & 0.10   & 0.50   & 0.80   & 1.00   & 1.50   & 2.00   \\ \hline
\multicolumn{1}{|c|}{\multirow{4}{*}{$\x^\sBB$}} & No    & 0.9380 & 0.8821 & 0.2791 & 0.0576 & 0.0155 & 0.0005 & 0.0000 \\
\multicolumn{1}{|c|}{}                           & LLL-P & 0.9777 & 0.9664 & 0.5703 & 0.1373 & 0.0347 & 0.0009 & 0.0000 \\
\multicolumn{1}{|c|}{}                           & SQRD  & 0.9705 & 0.9639 & 0.5802 & 0.1418 & 0.0353 & 0.0009 & 0.0000 \\
\multicolumn{1}{|c|}{}                           & LSP   & 0.9937 & 0.9892 & 0.6409 & 0.1507 & 0.0370 & 0.0010 & 0.0000 \\ \hline
\multicolumn{1}{|c|}{\multirow{6}{*}{$\x^\sRB$}} & No    & 0.9805 & 0.9636 & 0.7019 & 0.4690 & 0.3332 & 0.1510 & 0.0784 \\
\multicolumn{1}{|c|}{}                           & LLL-P & 0.9934 & 0.9891 & 0.8334 & 0.5433 & 0.3751 & 0.1596 & 0.0800 \\
\multicolumn{1}{|c|}{}                           & SQRD  & 0.9911 & 0.9887 & 0.8395 & 0.5476 & 0.3779 & 0.1615 & 0.0806 \\
\multicolumn{1}{|c|}{}                           & LSP   & 0.9980 & 0.9967 & 0.8569 & 0.5500 & 0.3758 & 0.1603 & 0.0788 \\
\multicolumn{1}{|c|}{}                           & GSP   & 0.9968 & 0.9935 & 0.8177 & 0.5433 & 0.3823 & 0.1720 & 0.0909 \\
\multicolumn{1}{|c|}{}                           & MSP   & 0.9975 & 0.9960 & 0.8435 & 0.5461 & 0.3780 & 0.1654 & 0.0843 \\ \hline
\end{tabular}
\end{table*}

\begin{table*}[]
\centering
\caption{SP of Babai points with column permutations, Type 2.}
\label{t:sp_perm_type2}
\begin{tabular}{|cc|ccccccc|}
\hline
\multicolumn{2}{|c|}{\multirow{2}{*}{$\P$}}              & \multicolumn{7}{c|}{$\sigma$}                                \\ \cline{3-9} 
\multicolumn{2}{|c|}{}                                   & 0.01   & 0.03   & 0.05   & 0.10   & 0.20   & 0.30   & 0.50   \\ \hline
\multicolumn{1}{|c|}{\multirow{4}{*}{$\x^\sBB$}} & No    & 0.9987 & 0.7520 & 0.3744 & 0.0463 & 0.0011 & 0.0001 & 0.0000 \\
\multicolumn{1}{|c|}{}                           & LLL-P & 1.0000 & 0.9439 & 0.6812 & 0.1510 & 0.0061 & 0.0003 & 0.0000 \\
\multicolumn{1}{|c|}{}                           & SQRD  & 1.0000 & 0.9353 & 0.6683 & 0.1572 & 0.0074 & 0.0004 & 0.0000 \\
\multicolumn{1}{|c|}{}                           & LSP   & 1.0000 & 0.9937 & 0.8578 & 0.2474 & 0.0097 & 0.0005 & 0.0000 \\ \hline
\multicolumn{1}{|c|}{\multirow{6}{*}{$\x^\sRB$}} & No    & 0.9995 & 0.9172 & 0.7671 & 0.4735 & 0.2187 & 0.1266 & 0.0616 \\
\multicolumn{1}{|c|}{}                           & LLL-P & 1.0000 & 0.9805 & 0.8872 & 0.6115 & 0.2960 & 0.1628 & 0.0727 \\
\multicolumn{1}{|c|}{}                           & SQRD  & 1.0000 & 0.9775 & 0.8850 & 0.6212 & 0.3131 & 0.1786 & 0.0796 \\
\multicolumn{1}{|c|}{}                           & LSP   & 1.0000 & 0.9975 & 0.9496 & 0.6769 & 0.3281 & 0.1750 & 0.0701 \\
\multicolumn{1}{|c|}{}                           & GSP   & 1.0000 & 0.9841 & 0.9027 & 0.6370 & 0.3110 & 0.1825 & 0.0914 \\
\multicolumn{1}{|c|}{}                           & MSP   & 1.0000 & 0.9968 & 0.9395 & 0.6590 & 0.3156 & 0.1741 & 0.0809 \\ \hline
\end{tabular}
\end{table*}

\subsection{Effects of column permutations on  changes in SP of $\x^\sRB$}\label{sec:numexp_spchanges_strategies}
In order to observe the effect of the column permutation strategies on the   SP of the $L_0$-regularized Babai point, we set $p\in \{0.3, 0.7\}$ and $\sigma \in \{0.2,0.5,1.5\}$. For each setting, we generate 1000 model matrices $\A$ for
each type. 
We count the number of experiment runs that each permutation strategy (LLL-P, SQRD, LSP, GSP and MSP) strictly increases, strictly decreases or does not change the theoretical SP, which is computed according to Theorem \ref{th:babaisp}.

It is easy to see that under all of our settings, 
the GSP and MSP algorithms do not decrease the SP, which is as expected. 
As $\sigma$ increases, LLL-P and LSP become more likely to decrease the SP. 
SQRD, on the other hand, can have higher chance of decreasing the SP when $\sigma$ increases from 0.2 to 0.5, although the trend becomes consistent with LLL-P and LSP when $\sigma$ is larger. 
It is also worth noting that when $\sigma = 0.2$ which is relatively small, 
LSP can still decrease the SP in some instances
despite that the average SP increases as presented in Section \ref{sec:numexp_avgsp_strategies}. 
This suggests that even if LSP has better average performance when $\sigma$ is small, 
we may still prefer MSP over LSP for the guaranteed SP increase .

\begin{table*}[]
\centering
\caption{Changes in SP of $\x^\sRB$ after applying column permutations, Type 1}
\label{t:theoretical_sp_change_gaussian_n20}
\begin{tabular}{|c|c|ccccc|ccccc|}
\hline
\multirow{2}{*}{$\sigma$} & \multirow{2}{*}{SP change} & \multicolumn{5}{c|}{$p=0.3$}                                                                                             & \multicolumn{5}{c|}{$p=0.7$}                                                                                             \\ \cline{3-12} 
                          &                            & \multicolumn{1}{c|}{LLL-P} & \multicolumn{1}{c|}{SQRD} & \multicolumn{1}{c|}{LSP} & \multicolumn{1}{c|}{GSP}  & MSP  & \multicolumn{1}{c|}{LLL-P} & \multicolumn{1}{c|}{SQRD} & \multicolumn{1}{c|}{LSP} & \multicolumn{1}{c|}{GSP}  & MSP  \\ \hline
\multirow{3}{*}{0.2}      & Strict increase            & \multicolumn{1}{c|}{994}   & \multicolumn{1}{c|}{953}  & \multicolumn{1}{c|}{997}     & \multicolumn{1}{c|}{1000} & 1000 & \multicolumn{1}{c|}{997}   & \multicolumn{1}{c|}{942}  & \multicolumn{1}{c|}{998}     & \multicolumn{1}{c|}{1000} & 1000 \\
                          & Strict decrease            & \multicolumn{1}{c|}{6}     & \multicolumn{1}{c|}{47}   & \multicolumn{1}{c|}{3}       & \multicolumn{1}{c|}{0}    & 0    & \multicolumn{1}{c|}{3}     & \multicolumn{1}{c|}{58}   & \multicolumn{1}{c|}{2}       & \multicolumn{1}{c|}{0}    & 0    \\
                          & No change                  & \multicolumn{1}{c|}{0}     & \multicolumn{1}{c|}{0}    & \multicolumn{1}{c|}{0}       & \multicolumn{1}{c|}{0}    & 0    & \multicolumn{1}{c|}{0}     & \multicolumn{1}{c|}{0}    & \multicolumn{1}{c|}{0}       & \multicolumn{1}{c|}{0}    & 0    \\ \hline
\multirow{3}{*}{0.5}      & Strict increase            & \multicolumn{1}{c|}{992}   & \multicolumn{1}{c|}{992}  & \multicolumn{1}{c|}{988}     & \multicolumn{1}{c|}{1000} & 1000 & \multicolumn{1}{c|}{994}   & \multicolumn{1}{c|}{986}  & \multicolumn{1}{c|}{988}     & \multicolumn{1}{c|}{1000} & 1000 \\
                          & Strict decrease            & \multicolumn{1}{c|}{8}   & \multicolumn{1}{c|}{8}  & \multicolumn{1}{c|}{12}     & \multicolumn{1}{c|}{0}    & 0    & \multicolumn{1}{c|}{6}    & \multicolumn{1}{c|}{14}   & \multicolumn{1}{c|}{12}      & \multicolumn{1}{c|}{0}    & 0    \\
                          & No change                  & \multicolumn{1}{c|}{0}     & \multicolumn{1}{c|}{0}    & \multicolumn{1}{c|}{0}       & \multicolumn{1}{c|}{0}    & 0    & \multicolumn{1}{c|}{0}     & \multicolumn{1}{c|}{0}    & \multicolumn{1}{c|}{0}       & \multicolumn{1}{c|}{0}    & 0    \\ \hline
\multirow{3}{*}{1.5}      & Strict increase            & \multicolumn{1}{c|}{823}   & \multicolumn{1}{c|}{854}  & \multicolumn{1}{c|}{802}      & \multicolumn{1}{c|}{1000} & 1000 & \multicolumn{1}{c|}{955}   & \multicolumn{1}{c|}{957}  & \multicolumn{1}{c|}{926}     & \multicolumn{1}{c|}{1000} & 1000 \\
                          & Strict decrease            & \multicolumn{1}{c|}{177}   & \multicolumn{1}{c|}{146}  & \multicolumn{1}{c|}{198}     & \multicolumn{1}{c|}{0}    & 0    & \multicolumn{1}{c|}{45}   & \multicolumn{1}{c|}{43}  & \multicolumn{1}{c|}{74}     & \multicolumn{1}{c|}{0}    & 0    \\
                          & No change                  & \multicolumn{1}{c|}{0}     & \multicolumn{1}{c|}{0}    & \multicolumn{1}{c|}{0}       & \multicolumn{1}{c|}{0}    & 0    & \multicolumn{1}{c|}{0}     & \multicolumn{1}{c|}{0}    & \multicolumn{1}{c|}{0}       & \multicolumn{1}{c|}{0}    & 0    \\ \hline
\end{tabular}
\end{table*}

\begin{table*}[]
\centering
\caption{ Changes in SP of $\x^\sRB$ after applying column permutations, Type 2}
\label{t:theoretical_sp_change_svd_n20}
\begin{tabular}{|c|c|ccccc|ccccc|}
\hline
\multirow{2}{*}{$\sigma$} & \multirow{2}{*}{SP change} & \multicolumn{5}{c|}{$p=0.3$}                                                                                             & \multicolumn{5}{c|}{$p=0.7$}                                                                                             \\ \cline{3-12} 
                          &                            & \multicolumn{1}{c|}{LLL-P} & \multicolumn{1}{c|}{SQRD} & \multicolumn{1}{c|}{LSP} & \multicolumn{1}{c|}{GSP}  & MSP  & \multicolumn{1}{c|}{LLL-P} & \multicolumn{1}{c|}{SQRD} & \multicolumn{1}{c|}{LSP} & \multicolumn{1}{c|}{GSP}  & MSP  \\ \hline
\multirow{3}{*}{0.2}      & Strict increase            & \multicolumn{1}{c|}{999}   & \multicolumn{1}{c|}{993}  & \multicolumn{1}{c|}{1000}     & \multicolumn{1}{c|}{1000} & 1000 & \multicolumn{1}{c|}{1000}   & \multicolumn{1}{c|}{991}  & \multicolumn{1}{c|}{1000}    & \multicolumn{1}{c|}{1000} & 1000 \\
                          & Strict decrease            & \multicolumn{1}{c|}{1}    & \multicolumn{1}{c|}{7}    & \multicolumn{1}{c|}{0}       & \multicolumn{1}{c|}{0}    & 0    & \multicolumn{1}{c|}{0}     & \multicolumn{1}{c|}{9}    & \multicolumn{1}{c|}{0}       & \multicolumn{1}{c|}{0}    & 0    \\
                          & No change                  & \multicolumn{1}{c|}{0}     & \multicolumn{1}{c|}{0}    & \multicolumn{1}{c|}{0}       & \multicolumn{1}{c|}{0}    & 0    & \multicolumn{1}{c|}{0}     & \multicolumn{1}{c|}{0}    & \multicolumn{1}{c|}{0}       & \multicolumn{1}{c|}{0}    & 0    \\ \hline
\multirow{3}{*}{0.5}      & Strict increase            & \multicolumn{1}{c|}{892}   & \multicolumn{1}{c|}{985}  & \multicolumn{1}{c|}{832}      & \multicolumn{1}{c|}{1000} & 1000 & \multicolumn{1}{c|}{989}   & \multicolumn{1}{c|}{999}  & \multicolumn{1}{c|}{988}     & \multicolumn{1}{c|}{1000} & 1000 \\
                          & Strict decrease            & \multicolumn{1}{c|}{108}   & \multicolumn{1}{c|}{15}  & \multicolumn{1}{c|}{168}     & \multicolumn{1}{c|}{0}    & 0    & \multicolumn{1}{c|}{11}   & \multicolumn{1}{c|}{1}  & \multicolumn{1}{c|}{12}     & \multicolumn{1}{c|}{0}    & 0    \\
                          & No change                  & \multicolumn{1}{c|}{0}     & \multicolumn{1}{c|}{0}    & \multicolumn{1}{c|}{0}       & \multicolumn{1}{c|}{0}    & 0    & \multicolumn{1}{c|}{0}     & \multicolumn{1}{c|}{0}    & \multicolumn{1}{c|}{0}       & \multicolumn{1}{c|}{0}    & 0    \\ \hline
\multirow{3}{*}{1.5}      & Strict increase            & \multicolumn{1}{c|}{157}    & \multicolumn{1}{c|}{186}    & \multicolumn{1}{c|}{21}       & \multicolumn{1}{c|}{1000} & 999  & \multicolumn{1}{c|}{414}    & \multicolumn{1}{c|}{625}   & \multicolumn{1}{c|}{107}       & \multicolumn{1}{c|}{1000} & 1000  \\
                          & Strict decrease            & \multicolumn{1}{c|}{843}   & \multicolumn{1}{c|}{814}  & \multicolumn{1}{c|}{979}     & \multicolumn{1}{c|}{0}    & 0    & \multicolumn{1}{c|}{586}   & \multicolumn{1}{c|}{375}  & \multicolumn{1}{c|}{893}     & \multicolumn{1}{c|}{0}    & 0    \\
                          & No change                  & \multicolumn{1}{c|}{0}     & \multicolumn{1}{c|}{0}    & \multicolumn{1}{c|}{0}       & \multicolumn{1}{c|}{0}    & 1    & \multicolumn{1}{c|}{0}     & \multicolumn{1}{c|}{0}    & \multicolumn{1}{c|}{0}       & \multicolumn{1}{c|}{0}    & 0    \\ \hline
\end{tabular}
\end{table*}

\section{Summary and Conclusion}
We have derived a formula for
the SP of the $L_0$-regularized Babai point $\x^\sRB$ under 
an $L_0$-RBILS problem setting to detect $\x^*$ in the linear model
\eqref{eq:lm}.
We have shown that 
with the regularization parameter $\lambda$ given by \eqref{eq:lambda},
$\x^\sRB$ is more preferable than the unregularized 
box-constrained Babai point $\x^\sBB$ 
in terms of success probability.
Furthermore, we have shown that the SP of $\x^\sRB$ 
monotonically increases with respect to the magnitude of the diagonal entries of the factor $\R$ obtained from the QR factorization of the model matrix $\A$, 
and that the ratio of 
the SP of $\x^\sRB$ to that of $\x^\sBB$
monotonically decreases as $p$ increases within its domain.
Additionally, a bound on the SP of $\x^\sRB$ has been derived, 
which acts as a lower or upper bound 
depending on the conditions satisfied. 

When columns of the model matrix $\A$ are re-ordered,
$\x^\sRB$ changes, so does its SP.
We have shown that applying the LLL-P permutation strategy guarantees 
to increase or decrease the SP depending on 
the conditions satisfied. 
Then three SP based permutation strategies have been proposed,
aiming to increase the SP of $\x^\sRB$. 
The first strategy LSP chooses the ordering greedily 
according to the local SP value at each level,
which produces the same ordering as the existing V-BLAST strategy. 
LSP can have great performance, however a major disadvantage 
is that it can decrease the SP.
The second strategy GSP, on the other hand, 
chooses the ordering greedily according to the global SP
and guarantees the SP not to decrease.
The third strategy MSP combines the advantages of LSP and GSP, 
having performance close to LSP when LSP performs well, 
while guaranteeing the SP not to decrease. 

The idea of the proposed approach for column permutations is quite general. It can be applied to other Babai estimation problems, such as the box constrained Babai estimation,  where each entry of the true integer parameter vector is 
uniformly distributed over $\calX/\{0\}$, i.e., the alphabet set $\calX$ without 0 included.


\section*{Acknowledgment}
We acknowledge the support of the Natural Sciences and Engineering Research Council of Canada (NSERC) [funding reference number 97463].

\appendix
\section*{Proofs of lemmas}
\subsection{Proof of Lemma \ref{le:akj}}\label{app:le:akj}
\begin{proof} 
It is obvious that $\alpha_k^{(1)} \geq -1/2$ and 
$\alpha_k^{(j)}$ is strictly decreasing when $j$ is increasing from 1 to $M$.

If $j_k \geq 2$, from \eqref{eq:alphakj} and \eqref{eq:j0}
\begin{align*}
\alpha_k^{(j_k-1)} 
 =  \frac{1}{2j_k-3}\frac{\lambda}{r_{kk}^2} - \frac{1}{2}(2j_k-3)  
  =  \frac{ 2j_k-1}{2j_k-3} \alpha_k^{(j_k)} + \frac{4j_k-4}{2j_k-3}  
  \geq -\frac{ 2j_k-1}{2j_k-3} +  \frac{4j_k-4}{2j_k-3}  
  =  1.
\end{align*}
Then  \eqref{eq:akjg1} follows by monotonicity.

If $j_k\leq M-1$, obviously the second inequality in \eqref{eq:akjl1} holds by the definition of $j_k$,
and in particular $\alpha_k^{(j_k+1)}<-1$. Then
\begin{align*}
\alpha_k^{(j_k)} 
  =  \frac{1}{2j_k-1}\frac{\lambda}{r_{kk}^2} - \frac{1}{2}(2j_k-1)  
  = \frac{ 2j_k+1}{2j_k-1} \alpha_k^{(j_k+1)}  - \frac{4j_k}{2j_k-1} 
  \leq   -\frac{2j_k+1}{ 2j_k-1} +  \frac{4j_k}{2j_k-1}  
  = 1,
\end{align*}
completing the proof.
\end{proof}

\subsection{Proof of Lemma \ref{le:ddF_properties}} \label{app:le:ddF_properties}
\begin{proof} 
To simplify notation, denote 
\begin{equation}\label{eq:theta}
\theta_\gamma := \frac{\blam}{2(2j_\gamma-1)\gamma} + \frac{1}{2}(2j_\gamma-1)\gamma  .
\end{equation}
In Lemma \ref{le:rkcts} we showed that $\rho(\gamma)$ and $\rho'(\gamma)$ is continuous on $(0,\infty)$. By \eqref{eq:dF}, the continuity of $F'(\zeta)$ on $(-\infty,\infty)$ can be directly obtained from these results, completing the proof for Lemma \ref{le:ddF_properties}-\ref{le:F1der}.

Let $\gamma = e^\zeta\in (0,\infty)$. 
Then, by \eqref{eq:ddF}, to prove our second claim, 
it is sufficient to prove that 
\al{
& F''(\ln \gamma) \to 0^+ \text{ as } \gamma \to 0^+,   \label{eq:ddFg1} \\
& F''(\ln \gamma) \to 0^- \text{ as } \gamma \to \infty.  \label{eq:ddFg2}
}
where 
\be \label{eq:F2der}
F''(\ln \gamma)= \frac{\rho''(\gamma)}{\rho(\gamma)}\gamma^2 + \frac{\rho'(\gamma)}{\rho(\gamma)} \gamma - \left( \frac{\rho'(\gamma)}{\rho(\gamma)} \gamma \right)^2
\ee

When $\gamma \in (\gamma^{(i-1)}, \gamma^{(i)})$ for any $i=1,\ldots,M$, 
from \eqref{eq:drkg}
the first derivative of $\rho(\gamma)$ is given by 
\al{
\rho'(\gamma)= \frac{2}{\sqrt{\pi}} \left[ \frac{M-j_\gamma}{M}pe^{-\gamma^2} + (1-p)(2j_\gamma-1)e^{-\theta_\gamma^2} \right], \label{eq:drho_thetas}
}
and the second derivative  is given by
\al{
  \rho''(\gamma) = \frac{2}{\sqrt{\pi}}\Bigg[ -\frac{2(M-j_\gamma)}{M}p\gamma e^{-\gamma^2}   
  + \frac{1}{2}(1-p)(2j_\gamma-1)   
     e^{-\theta_\gamma^2} \left( \frac{\blam^2}{(2j_\gamma-1)^2 \gamma^3} - (2j_\gamma-1)^2\gamma \right)\Bigg] \label{eq:ddrho_thetas}.
}

First we consider $\gamma\to 0^+$.
By Corollary \ref{cor:rbarlimit} we have 
\be \label{eq:rholim0}
\lim_{\gamma \to 0^+} \rho(\gamma) = 1-p.
\ee
From \eqref{eq:j0exp}, when $\gamma$ is small enough, we have
$\lim_{\gamma \to 0^+} j_\gamma = M.$
From \eqref{eq:theta} we have $\theta_\gamma \to \infty $, leading to $e^{-\theta_\gamma^2} \to 0$.
Therefore from \eqref{eq:drho_thetas} we observe that as $\gamma\to 0^+$,
\be \label{eq:drholim0}
\rho'(\gamma) \to 0^+.
\ee
It is easy to see that $\rho''(\gamma)>0$ when $\gamma$ is small enough. As $\gamma \to 0^+$, $e^{-\theta_\gamma^2}$ decays exponentially, which dominates over the explosion of $\frac{\blam^2}{(2j_\gamma - 1)^2 \gamma^3}$. Therefore from \eqref{eq:ddrho_thetas} we observe that as $\gamma\to 0^+$, 
\be \label{eq:ddrholim0}
\rho''(\gamma) \to 0^+
\ee
Combining \eqref{eq:rholim0}, \eqref{eq:drholim0} and \eqref{eq:ddrholim0}, from \eqref{eq:F2der} we can conclude that \eqref{eq:ddFg1} holds.

Now we consider
$\gamma\to \infty$.
By Corollary \ref{cor:rbarlimit} we have 
\be \label{eq:rholiminf}
\lim_{\gamma\to \infty} \rho(\gamma) = 1.
\ee
From \eqref{eq:j0exp}, we observe that when $\gamma$ is big enough, 
$j_\gamma = 1.$
Therefore, when $\gamma$ is big enough, we have
\al{
    \theta_\gamma & = \frac{\blam}{2\gamma} + \frac{1}{2}\gamma, \notag \\
    \rho'(\gamma)\gamma & =
    \frac{2}{\sqrt{\pi}}\left[ \frac{M-1}{M}p\gamma e^{-\gamma^2} 
    + (1-p)\gamma e^{-\theta_\gamma^2} \right],  \label{eq:drholiminf1} \\
    \rho''(\gamma)\gamma^2 
    & = \frac{2}{\sqrt{\pi}}
    \Bigg[ \frac{1}{2}(1-p)\left( \frac{\blam^2}{\gamma} - \gamma^3 \right)e^{-\theta_\gamma^2}  
 -\frac{2(M-1)}{M}p\gamma^3 e^{-\gamma^2} \Bigg], \notag 
}
leading to
\aln{
\rho''(\gamma)\gamma^2 + \rho'(\gamma)\gamma =&
 \frac{2}{\sqrt{\pi}}\Bigg[ \frac{M-1}{M}p(\gamma -2\gamma^3)e^{-\gamma^2}   
  + (1-p) \left(\gamma -\frac{1}{2}\gamma^3 + \frac{\blam^2}{\gamma}\right)  e^{-\theta_\gamma^2}\Bigg].
}
When $\gamma$ is large enough, obviously 
$\rho''(\gamma)\gamma^2 + \rho'(\gamma)\gamma<0$.
Note that $e^{-\theta_\gamma^2}$ decays at an exponential rate with respect to $\gamma$.
Since the exponential decay is much faster than polynomial explosion,
the above limit gives
\be \label{eq:sumliminf}
\rho''(\gamma)\gamma^2 + \rho'(\gamma)\gamma \to 0^-.
\ee
From \eqref{eq:drholiminf1}, it is easy to see that when $\gamma\to\infty$,
\be \label{eq:drholiminf}
\rho'(\gamma)\gamma \to 0^+.
\ee
Combining \eqref{eq:rholiminf}, \eqref{eq:sumliminf} and \eqref{eq:drholiminf}, from \eqref{eq:F2der}  
it is easy to see that \eqref{eq:ddFg2} holds.
This completes our proof for Lemma \ref{le:ddF_properties}-\ref{le:F2derend}.

When $M=1$, $j_\gamma = 1$ for any $\gamma > 0$. 
From \eqref{eq:ddrho_thetas} we can obtain that
$\rho''(\gamma)$ is continuous on $(0,\infty)$. 
Using this and the continuity of 
$\rho'(\gamma)$ and $\rho(\gamma)$ shown in Lemma \ref{le:rkcts}, 
from \eqref{eq:F2der} we can obtain the continuity of $F''(\zeta)$ with $\zeta=\ln \gamma$ 
on $(-\infty, \infty)$, concluding the proof for Lemma \ref{le:ddF_properties}-\ref{le:M1F2der}.

For any $i=1,\ldots,M-1$, when $\zeta \in (\ln\gamma^{(i-1)}, \ln\gamma^{(i)})$, 
$\gamma = e^\zeta \in (\gamma^{(i-1)}, \gamma^{(i)})$. 
By \eqref{eq:jkpw} the value of $j_\gamma$ is fixed for $\gamma\in (\gamma^{(i-1)}, \gamma^{(i)})$. 
Therefore by \eqref{eq:ddrho_thetas}, $\rho''(\gamma)$ is continuous in $(\gamma^{(i-1)}, \gamma^{(i)})$.
By Lemma \ref{le:rkcts}, $\rho(\gamma)$ and $\rho'(\gamma)$ are continuous 
in $(0,\infty)$.
Therefore $F''(\zeta)$ is continuous on $(\ln\gamma^{(i-1)}, \ln\gamma^{(i)})$. 

Using \eqref{eq:ddF} and the continuity of $\rho(\gamma)$ and $\rho'(\gamma)$ 
from Lemma \ref{le:rkcts}, we see that 
to show that $\lim_{\zeta \to \ln\gamma^{(i)-}} F''(\zeta) < \lim_{\zeta \to \ln\gamma^{(i)+}} F''(\zeta)$, 
it suffices to show
\begin{equation}\label{eq:ddrho_jump}
    \lim_{\gamma\to \gamma^{(i)-}} \rho''(\gamma) - \lim_{\gamma\to\gamma^{(i)+}} \rho''(\gamma) < 0.
\end{equation}
In the following we derive an inequality equivalent to 
\eqref{eq:ddrho_jump} by finding the expressions of the two limits.

Using \eqref{eq:jkpw}, from \eqref{eq:theta} we can easily verify that 
\aln{
    & \lim_{\gamma\to \gamma^{(i)-}}\theta_\gamma = \lim_{\gamma\to \gamma^{(i)+}}\theta_\gamma  
    =  \frac{1}{2} \left( \sqrt{\frac{\blam(2(M-i)-1)}{2(M-i)+1}} + \sqrt{\frac{\blam(2(M-i)+1)}{2(M-i)-1}} \right).
}
Therefore, when $\gamma\to \gamma^{(i)}$ from either side,
\begin{equation}\label{eq:thetasqlim}
    \theta_\gamma^2 \to \frac{4\blam(M-i)^2}{4(M-i)^2-1}.
\end{equation}
Using \eqref{eq:ddrho_thetas}, \eqref{eq:jkpw} and \eqref{eq:thetasqlim},
from \eqref{eq:ddrho_jump}
we obtain the following equivalent inequality:
\aln{
    &\left( \frac{2(i-1)}{M} - \frac{2i}{M} \right) p\gamma^{(i)} e^{-(\gamma^{(i)})^2} \\    
    &+ \frac{1}{2}(1-p)(2(M-i)+1)e^{-\frac{4\blam(M-i)^2}{4(M-i)^2-1}}  
      \times \left[ \frac{\blam^2}{(2(M-i)+1)^2 (\gamma^{(i)})^3} - (2(M-i)+1)^2 \gamma^{(i)} \right] \\
    &- \frac{1}{2}(1-p)(2(M-i)-1)e^{-\frac{4\blam(M-i)^2}{4(M-i)^2-1}}  
       \times \left[ \frac{\blam^2}{(2(M-i)-1)^2 (\gamma^{(i)})^3} - (2(M-i)-1)^2 \gamma^{(i)} \right]  
     <   0.
}
After we substitute \eqref{eq:giexp} into the above inequality 
and do some algebraic manipulations, it becomes
\begin{equation*}
    \frac{p}{2M}e^{-\frac{\blam}{4(M-i)^2-1}} - 4(M-i)^2(1-p)e^{-\frac{4\blam(M-i)^2}{4(M-i)^2-1}} < 0 .
\end{equation*}
Using $\frac{1-p}{p/(2M)} = e^\blam$ (see \eqref{eq:lambda}), we can rewrite the above inequality as
\begin{equation*}
    4(M-i)^2 > e^{ \blam\left[ \frac{4(M-i)^2}{4(M-i)^2-1} -\frac{1}{4(M-i)^2-1} -1 \right] } = 1.
\end{equation*}
Obviously, the above inequality holds for any $i=1,\ldots,M-1$.
This concludes the proof for Lemma \ref{le:ddF_properties}-\ref{le:M2F2der}.

Finally, we prove the claim in Lemma \ref{le:ddF_properties}-\ref{le:dFmono}.
Recall that  $\ln\mu_1$ is the smallest root of $F''(\zeta)$.
When $M=1$, by Lemma \ref{le:ddF_properties}-\ref{le:F2derend} and  \ref{le:ddF_properties}-\ref{le:M1F2der}  
it follows that $F''(\zeta)>0$ for any $\zeta \in (-\infty, \ln\mu_1)$.
When $M\geq 2$, by the first part of  Lemma \ref{le:ddF_properties}-\ref{le:M2F2der},
we observe that $F''(\zeta)$ does not change sign over the interval 
$(\ln \gamma^{(i-1)},\ln \gamma^{(i)}) \cap (-\infty, \ln \mu_1)$,
otherwise $\ln\mu_1$ would not be the smallest root of $F''(\zeta)$.
Then, by the first part of Lemma \ref{le:ddF_properties}-\ref{le:F2derend} and the second part of 
Lemma \ref{le:ddF_properties}-\ref{le:M2F2der}
we can conclude that $F''(\zeta)>0$ for any 
$\zeta \in  (\ln \gamma^{(i-1)},\ln \gamma^{(i)}) \cap (-\infty, \ln \mu_1)$. 
Therefore, by Lemma \ref{le:ddF_properties}-\ref{le:F1der}
we can conclude that $F'(\zeta)$ is strictly increasing on 
$(-\infty, \ln \mu_1)$.

A symmetrical argument can be used to show that 
$F'(\zeta)$ is strictly decreasing on $(\ln \mu_2, \infty)$.
This concludes our proof for Lemma \ref{le:ddF_properties}-\ref{le:dFmono}.
\end{proof}

\subsection{Proof of Lemma \ref{le:phimono}}\label{app:le:phimono}
\begin{proof}
From \eqref{eq:Phi} and \eqref{eq:dF}, we obtain
\al{
\frac{\partial}{\partial \gamma}\Phi(\gamma,\beta)  &= \frac{1}{\gamma} \rho(\gamma)  \rho\bigg(\frac{\beta}{\gamma}\bigg)  \left[ \frac{\gamma \rho'(\gamma)}{ \rho(\gamma)} - \frac{\left(\frac{\beta}{\gamma}\right)\rho'\left(\frac{\beta}{\gamma}\right) }{ \rho \left(\frac{\beta}{\gamma}\right)} \right]  
 = \frac{1}{\gamma} \rho(\gamma)  \rho\bigg(\frac{\beta}{\gamma}\bigg) \left[ F'(\ln \gamma) - F'\left( \ln \frac{\beta}{\gamma} \right) \right]. \label{eq:dPhizb} 
}
where the function $F'$ is as in \eqref{eq:dF}.
To prove our claim, it suffices to show that $F'(\ln\gamma) < F'(\ln(\beta/\gamma))$ when either $\beta/\mu_1 < \gamma < \sqrt{\beta}$ or $\sqrt{\beta} < \gamma < \beta/\mu_2$. 

Since the function $\ln(\cdot)$ is strictly increasing on $(0,\infty)$,
by Lemma \ref{le:ddF_properties}-\ref{le:dFmono}
we have that $F'(\ln \gamma)$ is strictly increasing when $\gamma \in (0,\mu_1)$ and strictly decreasing when $\gamma \in (\mu_2, +\infty)$. 

First, we assume $\beta/\mu_1 < \gamma < \sqrt{\beta}$
holds. 
Then we have that $\gamma<\beta/\gamma<\mu_1$. 
Since 
$F'(\ln\gamma)$ is strictly increasing on $(0,\mu_1)$,
it is easy to obtain $F'(\ln \gamma) < F'(\ln (\beta/\gamma))$. 

Next we assume $\sqrt{\beta} < \gamma < \beta/\mu_2$
holds. 
Then we have that 
$\mu_2 < \beta/\gamma < \gamma$.
Since 
$F'(\ln\gamma)$ is strictly decreasing on $(\mu_2,+\infty)$,
we have $F'(\ln \gamma) < F'(
\ln(\beta/\gamma))$.
\end{proof}

\bibliographystyle{IEEEtran}
\bibliography{ref} 

\begin{thebibliography}{10}
\providecommand{\url}[1]{#1}
\csname url@samestyle\endcsname
\providecommand{\newblock}{\relax}
\providecommand{\bibinfo}[2]{#2}
\providecommand{\BIBentrySTDinterwordspacing}{\spaceskip=0pt\relax}
\providecommand{\BIBentryALTinterwordstretchfactor}{4}
\providecommand{\BIBentryALTinterwordspacing}{\spaceskip=\fontdimen2\font plus
\BIBentryALTinterwordstretchfactor\fontdimen3\font minus
  \fontdimen4\font\relax}
\providecommand{\BIBforeignlanguage}[2]{{%
\expandafter\ifx\csname l@#1\endcsname\relax
\typeout{** WARNING: IEEEtran.bst: No hyphenation pattern has been}%
\typeout{** loaded for the language `#1'. Using the pattern for}%
\typeout{** the default language instead.}%
\else
\language=\csname l@#1\endcsname
\fi
#2}}
\providecommand{\BIBdecl}{\relax}
\BIBdecl

\bibitem{ZhuG11}
H.~Zhu and G.~Giannakis, ``Exploiting sparse user activity in multiuser
  detection,'' \emph{IEEE Trans. Commun.}, vol.~59, no.~2, pp. 454--465, 2011.

\bibitem{Mic01}
D.~Micciancio, ``The hardness of the closest vector problem with
  preprocessing,'' \emph{IEEE Trans. Inf. Theory}, vol.~47, no.~3, pp.
  1212--1215, 2001.

\bibitem{Van81}
P.~van Emde~Boas, ``Another np-complete problem and the complexity of computing
  short vectors in a lattice,'' \emph{Tecnical Report, Department of
  Mathmatics, University of Amsterdam}, 1981.

\bibitem{HanPS11}
G.~Hanrot, X.~Xavier~Pujol, and D.~Stehl{\'e}, ``Algorithms for the shortest
  and closest lattice vector problems,'' in \emph{IWCC'11 Proceedings of the
  Third international conference on Coding and cryptology}.\hskip 1em plus
  0.5em minus 0.4em\relax Springer-Verlag Berlin, Heidelberg, 2011, pp.
  159--190.

\bibitem{AnjCK14}
M.~F. Anjos, X.-W. Chang, and W.-Y. Ku, ``Lattice preconditioning for the real
  relaxation branch-and-bound approach for integer least squares problems,''
  \emph{Journal of Global Optimization}, vol.~59, no. 2-3, pp. 227--242, 2014.

\bibitem{FinP85}
U.~Fincke and M.~Pohst, ``Improved methods for calculating vectors of short
  length in a lattice, including a complexity analysis,'' \emph{Math. Comput.},
  vol.~44, no. 170, pp. 463--471, 1985.

\bibitem{SchE94}
C.~Schnorr and M.~Euchner, ``Lattice basis reduction: improved practical
  algorithms and solving subset sum problems,'' \emph{Math Program}, vol.~66,
  pp. 181--191, 1994.

\bibitem{DamGC03}
M.~O. Damen, H.~E. Gamal, and G.~Caire, ``On maximum likelihood detection and
  the search for the closest lattice point,'' \emph{IEEE Trans. Inf. Theory},
  vol.~49, no.~10, pp. 2389--2402, 2003.

\bibitem{ChaH08}
X.-W. Chang and Q.~Han, ``Solving box-constrained integer least squares
  problems,'' \emph{IEEE Trans. Wireless Commun.}, vol.~7, no.~1, pp. 277--287,
  2008.

\bibitem{ChaG10}
X.-W. Chang and G.~Golub, ``Solving ellipsoid constrained integer least squares
  problems,'' \emph{SIAM Journal on Matrix Analysis and Applications}, vol.~31,
  no.~3, pp. 1071--1089, 2010.

\bibitem{BarV14}
S.~Barik and H.~Vikalo, ``Sparsity-aware sphere decoding: Algorithms and
  complexity analysis,'' \emph{IEEE Transactions on Signal Processing},
  vol.~62, no.~9, pp. 2212--2225, 2014.

\bibitem{MonBD13}
F.~Monsees, C.~Bockelmann, and A.~Dekorsy, ``Compressed sensing soft activity
  processing for sparse multi-user systems,'' in \emph{2013 IEEE Globecom
  Workshops (GC Wkshps)}.\hskip 1em plus 0.5em minus 0.4em\relax IEEE, 2013,
  pp. 247--251.

\bibitem{LenLL82}
A.~Lenstra, H.~Lenstra, and L.~Lov{\'a}sz, ``Factoring polynomials with
  rational coefficients,'' \emph{Math. Ann.}, vol. 261, no.~4, pp. 515--534,
  1982.

\bibitem{FosGVW99}
G.~J. Foscini, G.~D. Golden, R.~A. Valenzuela, and P.~W. Wolniansky,
  ``Simplified processing for high spectral efficiency wireless communication
  employing multi-element arrays,'' \emph{IEEE J. Sel. Areas Commun.}, vol.~17,
  no.~11, pp. 1841--1852, 1999.

\bibitem{WubBRKK01}
D.~W\"{u}bben, R.~Bohnke, J.~Rinas, V.~Kuhn, and K.~Kammeyer, ``Efficient
  algorithm for decoding layered space-time codes,'' \emph{Electron. Lett.},
  vol.~37, no.~22, pp. 1348--1350, 2001.

\bibitem{Bab86}
L.~Babai, ``On lovasz lattice reduction and the nearest lattice point
  problem,'' \emph{Combinatorica}, vol.~6, no.~1, pp. 1--13, 1986.

\bibitem{ChaWX13}
X.-W. Chang, J.~Wen, and X.~Xie, ``Effects of the {LLL} reduction on the
  success probability of the babai point and on the complexity of sphere
  decoding,'' \emph{IEEE Trans. Inf. Theory}, vol.~59, no.~8, pp. 4915--4926,
  2013.

\bibitem{WenC17}
J.~Wen and X.-W. Chang, ``The success probability of the {B}abai point
  estimator and the integer least squares estimator in box-constrained integer
  linear models,'' \emph{IEEE Trans. Inf. Theory}, vol.~63, pp. 631--648, 2017.

\bibitem{ChoSDRK17}
J.~W. Choi, B.~Shim, Y.~Ding, B.~Rao, and D.~I. Kim, ``Compressed sensing for
  wireless communications: Useful tips and tricks,'' \emph{IEEE Commun. Surv.
  Tutor.}, vol.~19, no.~3, pp. 1527--1548, 2017.

\bibitem{SchD11}
H.~F. Schepker and A.~Dekorsy, ``Sparse multi-user detection for cdma
  transmission using greedy algorithms,'' in \emph{2011 8th International
  Symposium on Wireless Communication Systems}, pp. 291--295.

\bibitem{SchD12}
------, ``Compressive sensing multi-user detection with block-wise orthogonal
  least squares,'' in \emph{2012 IEEE 75th Vehicular Technology Conference (VTC
  Spring)}, pp. 1--5.

\bibitem{KnoMBWPD13}
B.~Knoop, F.~Monsees, C.~Bockelmann, D.~Wuebben, S.~Paul, and A.~Dekorsy,
  ``Sparsity-aware successive interference cancellation with practical
  constraints,'' in \emph{WSA 2013; 17th International ITG Workshop on Smart
  Antennas}, 2013, pp. 1--8.

\bibitem{PantLA14}
J.~K. Pant, W.-S. Lu, and A.~Antoniou, ``New improved algorithms for
  compressive sensing based on $\ell_{p}$ norm,'' \emph{IEEE Transactions on
  Circuits and Systems II: Express Briefs}, vol.~61, no.~3, pp. 198--202, 2014.

\bibitem{LeeCS16}
J.~Lee, J.~W. Choi, and B.~Shim, ``Sparse signal recovery via tree search
  matching pursuit,'' \emph{Journal of Communications and Networks}, vol.~18,
  no.~5, pp. 699--712, 2016.

\bibitem{AhnSL17}
J.~Ahn, B.~Shim, and K.~B. Lee, ``Sparsity-aware ordered successive
  interference cancellation for massive machine-type communications,''
  \emph{IEEE Wireless Communications Letters}, vol.~7, no.~1, pp. 134--137,
  2017.

\bibitem{WenPYYL17}
F.~Wen, L.~Pei, Y.~Yang, W.~Yu, and P.~Liu, ``Efficient and robust recovery of
  sparse signal and image using generalized nonconvex regularization,''
  \emph{IEEE Transactions on Computational Imaging}, vol.~3, no.~4, pp.
  566--579, 2017.

\bibitem{SelF17}
I.~Selesnick and M.~Farshchian, ``Sparse signal approximation via nonseparable
  regularization,'' \emph{IEEE Transactions on Signal Processing}, vol.~65,
  no.~10, pp. 2561--2575, 2017.

\bibitem{SouL17}
N.~Souto and H.~Lopes, ``Efficient recovery algorithm for discrete valued
  sparse signals using an admm approach,'' \emph{IEEE Access}, vol.~5, pp.
  9562--19\,569, 2017.

\bibitem{CirML18}
A.~C. Cirik, N.~Mysore~Balasubramanya, and L.~Lampe, ``Multi-user detection
  using admm-based compressive sensing for uplink grant-free noma,'' \emph{IEEE
  Wireless Communications Letters}, vol.~7, no.~1, pp. 46--49, 2018.

\bibitem{FukNS19}
L.~Fukshansky, D.~Needell, and B.~Sudakov, ``An algebraic perspective on
  integer sparse recovery,'' \emph{Applied Mathematics and Computation}, vol.
  340, pp. 31--42, 2019.

\bibitem{Hay20}
H.~R., \emph{Studies on Discrete-Valued Vector Reconstruction from
  Underdetermined Linear Measurements}.\hskip 1em plus 0.5em minus 0.4em\relax
  PhD thesis, Department of Systems Science, Graduate School of Informatics,
  Kyoto University, 2020.

\bibitem{LanPST20}
J.-H. Lange, M.~Pfetsch, B.~Seib, and A.~Tillmann, ``Sparse recovery with
  integrality constraints,'' \emph{Discrete Applied Mathematics}, vol. 283, pp.
  346--366, 2020.

\bibitem{CheZCZ21}
Z.~Chen, S.~Zhong, C.~J., and Y.~Zhao, ``Deeppursuit: Uniting classical wisdom
  and deep rl for sparse recovery,'' in \emph{55th Asilomar Conference on
  Signals, Systems, and Computers, 2021}, pp. 1361--1366.

\bibitem{ChaP07}
X.-W. Chang and C.~Paige, ``Euclidean distances and least squares problems for
  a given set of vectors,'' \emph{Applied Numerical Mathematics}, vol.~57,
  no.~1, pp. 1240--1244, 2007.

\bibitem{GolV13}
G.~Golub and C.~Van~Loan, ``Matrix computations, 4th,'' \emph{Johns Hopkins},
  2013.

\bibitem{ChaP18}
X.-W. Chang and V.~Prevost, ``Success probability of a suboptimal solution to
  the sparse map detection,'' in \emph{2018 IEEE International Symposium on
  Information Theory (ISIT)}, pp. 66--70.

\end{thebibliography}
   
\end{document}